\documentclass[12pt]{iopart}
\usepackage{iopams,mathptm,times,graphicx}

\makeatletter

\newcounter{theorem}
\@addtoreset{theorem}{section}
\renewcommand\thetheorem{\arabic{section}.\arabic{theorem}}

\newenvironment{lemma}{\par\medskip\noindent\begingroup{\bf Lemma
             \stepcounter{theorem}\thetheorem.}\ \itshape
             \def\@currentlabel{\thetheorem}}{\endgroup\par\medskip}
\newenvironment{theorem}{\par\medskip\noindent\begingroup{\bf Theorem
             \stepcounter{theorem}\thetheorem.}\ \itshape
             \def\@currentlabel{\thetheorem}}{\endgroup\par\medskip}

\newenvironment{remark}{\par\medskip\noindent\begingroup{\bf Remark
             \stepcounter{theorem}\thetheorem.}\
             \def\@currentlabel{\thetheorem}}{\endgroup\par\medskip}

\newenvironment{proof}{\par\noindent{\bf Proof.}
}{\proofbox\par\medskip}

\newenvironment{example}{\par\noindent{\bf Example.} }{\proofbox\par\medskip}

\def\proofbox{\hfill{\ensuremath\Box}}

\newdimen\LENB \newdimen\LENW \newdimen\THI
\newdimen\LENWH \newdimen\LENTOT \newcount\N
\def\vbrknlnele#1#2#3{
  \LENB=#1pt \LENW=#2pt \THI=#3pt
  \LENWH=\LENW \divide\LENWH by 2
  \LENTOT=\LENB \advance\LENTOT by \LENW
  \vbox to \LENTOT{
    \vbox to \LENWH{}
    \nointerlineskip
    \vbox to \LENB{\hbox to \THI{\vrule width \THI height \LENB}}
    \nointerlineskip
    \vbox to \LENWH{}
  }}

\def\vbrknln#1{
  \N=#1
  \vcenter{
    \vbox{
      \loop\ifnum\N>0
        \vbox to 4pt{\vbrknlnele{2}{2}{0.1}}
        \nointerlineskip
        \advance\N by -1
      \repeat
  }}}

\def\hbrknlnele#1#2#3{
  \LENB=#1pt \LENW=#2pt \THI=#3pt
  \LENTOT=\LENB \advance\LENTOT by \LENW
  \vcenter{
    \vbox to \THI{
      \hbox to \LENTOT{
        \hfil
        \vrule width \LENB height \THI
        \hfil}
  }}}

\makeatletter

\usepackage{ulem}

\def\journal#1&#2,{\begingroup \let\journal=\dummyjournal
               \it #1\unskip~\bf\ignorespaces #2\rm,\endgroup}
\def\dummyjournal{\errmessage{Reference foul up: nested \journal macros}}

\eqnobysec

\def\eqref#1{(\ref{#1})}
\begin{document}
\title[On the $\tau$-functions of 
the Degasperis-Procesi equation]
  {On the $\tau$-functions of 
the Degasperis-Procesi equation}
\author{Bao-Feng Feng$^1$, Ken-ichi Maruno$^1$ and 
Yasuhiro Ohta$^{2}$ 
}
\address{$^1$~Department of Mathematics,
The University of Texas-Pan American,
Edinburg, TX 78539-2999
}
\address{$^2$~Department of Mathematics,
Kobe University, Rokko, Kobe 657-8501, Japan
}
\eads{\mailto{feng@utpa.edu}, \mailto{kmaruno@utpa.edu} and 
\mailto{ohta@math.kobe-u.ac.jp}} 

\date{\today}
\def\submitto#1{\vspace{28pt plus 10pt minus 18pt}
     \noindent{\small\rm To be submitted to : {\it #1}\par}}

\begin{abstract}
The DP equation is investigated from the point of view of 
determinant-pfaffian identities. 
The reciprocal link 
between the Degasperis-Procesi (DP) equation and 
the pseudo 3-reduction of the $C_{\infty}$ two-dimensional Toda system 
is used to construct 
the $N$-soliton solution of the DP equation. 
The $N$-soliton solution of the DP 
equation is presented in the form of pfaffian through a hodograph 
(reciprocal) transformation.   
The bilinear equations, the identities between determinants and pfaffians, and 
the $\tau$-functions of the DP equation are
 obtained from the pseudo 3-reduction of the $C_{\infty}$ 
two-dimensional Toda system. 
\par
\kern\bigskipamount\noindent
\today
\end{abstract}

\kern-\bigskipamount

\submitto{\JPA}

\section{Introduction}
In this article, we investigate the $N$-soliton 
solution of the Degasperis-Procesi (DP) equation~\cite{DP}
\begin{equation}
u_t+3\kappa^3u_x- u_{txx}+4uu_x=3u_xu_{xx}+uu_{xxx}\,,\label{DP-eq}
\end{equation}
which has received much attention in recent years. 
Since its discovery in 1999, 
many interesting mathematical properties of the DP equation have been
found~\cite{DHH,Hone-Wang,Lundmark,Lundmark2,Matsuno-DP1,
Matsuno-DP2,Constantin}. 
Matsuno derived the $N$-soliton solution of the DP equation 
when $\kappa\neq 0$ through the $\tau$-function of the CKP 
hierarchy~
\cite{Matsuno-DP1,Matsuno-DP2}.

In recent years, 
we proposed integrable discretizations of 
various integrable systems such as the Camassa-Holm (CH) equation
~\cite{dCH,dCH2}, 
the short wave limit of the CH equation~\cite{sCH}, 
the short pulse equation~\cite{SP}, 
the WKI elastic beam equation and the Dym equation~\cite{WKI}. 
Since these equations, as well as the DP equation, are
transformed to some integrable systems through hodograph (reciprocal) 
transformations, constructing integrable discretizations is not 
an easy task. 
In our past studies, we constructed the integrable discretizations 
of these equations by using the Hirota's bilinear approach and 
an approach based on discrete differential geometry. 
In our studies, we found the key of 
discretizations for these equations is discretizations of 
hodograph (reciprocal) transformations. 
Recently we found that integrable 
discretizations of some of these systems have a geometric formulation 
which verifies the discrete differential geometric 
meaning of discrete hodograph (reciprocal) 
transformations. 

The objective of this paper is to establish the useful formulation 
which can 
be used in the integrable discretization of the DP equation.   
Although the DP equation is very similar to the CH equation, 
there is an obstacle for constructing a discrete analogue of the DP
equation. Matsuno derived the $N$-soliton solutions of the DP equation 
in the form of pfaffian through the CKP equation, 
but bilinear equations (or pfaffian
identities) of the DP equation is still unclear. To construct 
a discrete analogue of the DP equation through the Hirota's bilinear
approach, we need to understand bilinear identities of pfaffians which 
consist of the DP equation. Moreover, we need to understand clearly 
how to
obtain the DP's $\tau$-functions from the ones of 
the KP hierarchy and the
two-dimensional Toda system. 
Thus we investigate the DP equation from the point of view of 
determinant-pfaffian identities. 
In the same motivation, we recently investigated the reduced Ostrovsky
equation~\cite{Ost}. It is known that the reduced Ostrovsky equation can be
obtained as a short wave limit of the DP equation~\cite{MatsunoPLA}. 

In this paper, we establish 
the reciprocal link between 
the DP equation and the pseudo 3-reduction of the  
$C_{\infty}$ two-dimensional Toda system and investigate 
the relations of $\tau$-functions. 
Using this reciprocal link and the relations of $\tau$-functions, 
we construct the $N$-soliton solution of the DP 
equation in the form of pfaffian.   
The bilinear equations, the identities between determinants 
and pfaffians, and 
the $\tau$-functions of the DP equation are
easily obtained from the pseudo 3-reduction of the 
$C_{\infty}$ two-dimensional Toda system. 

Note that 
Matsuno obtained 
the $N$-soliton solution through the $N$-soliton 
solution of the CKP hierarchy. 
Since the CKP hierarchy and the $C_{\infty}$ 
2D Toda system share the same $\tau$-function, it is natural to 
obtain the same soliton solution by both methods. 
However, by using $C_{\infty}$ 2D Toda system we can 
include a negative time variable in the $\tau$-functions, 
so there is an advantage for  
obtaining the bilinear equations of determinants and 
the relations of determinants and pfaffians. 

\section{The DP equation and the pseudo 3-reduction of the
 $C_{\infty}$ 2D-Toda system}

The two-dimensional Toda (2D-Toda) system of $A_{\infty}$-type, 
which is also called the Toda field equation  or the two-dimensional
Toda lattice, 
is given as follows~\cite{Leznov-Saveliev,Mikhailov1,Mikhailov2,Mikhailov3}:
\begin{equation}
\frac{\partial^2\theta_n}{\partial x_1 \partial x_{-1}}
=-\sum_{m\in \mathbb{Z}}a_{n,m}e^{-\theta_m}\,, 
\quad n\in \mathbb{Z}\label{2dToda}
\end{equation}
where the matrix $A=(a_{n,m})$ is the transpose of the Cartan matrix for
the infinite dimensional Lie algebra $A_{\infty}$, 
which is the infinite tridiagonal matrix
\begin{equation}
A=
\left[
\begin{array}{cccccccc}
\ddots & \ddots & \ddots & & & & & \\
       &  -1 & 2 & -1 & & & & \\
       &     & -1 & 2 & -1 & & & \\
       &     &    &  \ddots & \ddots & \ddots & & \\
       &     &    &         &  -1    &  2    & -1 & \\
      &      &    &         &        &  \ddots & \ddots & \ddots     
\end{array}
\right]\,.\label{cartan-A}
\end{equation}

The $A_\infty$ 2D-Toda system (\ref{2dToda}) with (\ref{cartan-A})
may be written as 
\begin{equation}
\frac{\partial^2\theta_n}{\partial x_1 \partial x_{-1}}
=e^{-\theta_{n-1}}-2e^{-\theta_{n}}+e^{-\theta_{n+1}}\,.  
\end{equation}
The $A_\infty$ 
2D-Toda system (\ref{2dToda}) with (\ref{cartan-A}) is transformed
into 
the bilinear equation
\begin{equation}
-\left(\frac{1}{2}D_{x_1}D_{x_{-1}}-1\right)\tau_n\cdot \tau_n=\tau_{n-1}\tau_{n+1}\,,\label{2dtl-bilinear}
\end{equation}
through the dependent variable transformation
\begin{equation}
\theta_n=-\ln \frac{\tau_{n+1}\tau_{n-1}}{\tau_n^2}\,.
\end{equation}
Here $D_{x}$ is the Hirota $D$-operator which is defined 
as
\begin{equation}
D_x^na(x)\cdot b(x)=\left(\partial_x-\partial_{x'}\right)^na(x)b(x')|_{x=x'}\,.
\end{equation}

\begin{lemma}[Ueno-Takasaki\cite{Ueno-Takasaki},
Babich-Matveev-Sall\cite{Matveev},
Hirota\cite{HirotaBook},Nimmo-Willox\cite{Nimmo-Willox}]
The bilinear equation (\ref{2dtl-bilinear}) 
and other bilinear equations of the members of 
the 2D-Toda lattice hierarchy have the 
following Gram-type determinant solution: 
\begin{equation*}
\tau_n={\rm det}\left(
\psi_{i,j}^{(n)} \right)_{1\leq i,j \leq M}\,,
\end{equation*}
where
\begin{equation*}
\psi_{i,j}^{(n)}=c_{i,j}+(-1)^n\int_{-\infty}^{x_1}\varphi_i^{(n)}
\hat{\varphi}_j^{(-n)}dx_1\,.
\end{equation*}
Here $c_{i,j}$ are constants,
$\varphi_{i,j}^{(n)}$ and $\hat{\varphi}_{i,j}^{(n)}$ must satisfy 
$\frac{\partial \varphi_i^{(n)}}{\partial x_k}=\varphi_i^{(n+k)}$
and 
$\frac{\partial \hat{\varphi}_i^{(n)}}{\partial x_k}
=(-1)^{k-1}\hat{\varphi}_i^{(n+k)}$ 
for $k=\pm 1, \pm 2, \pm 3, \cdots$.

For example, the following linear independent set of functions 
$\{\varphi_i^{(n)},\hat{\varphi}_j^{(n)}\}$ for $i,j=1,2,\cdots, M$ 
gives $M$-soliton solution of the $A_{\infty}$ 2D-Toda system: 
$$
\varphi_i^{(n)}= p_i^ne^{\xi_i}\,,\quad 
\hat{\varphi}_j^{(n)}=q_j^ne^{\eta_j}\,,
$$
where $\xi_i=p_ix_1+\frac{1}{p_i}x_{-1}+p_i^2x_2+\frac{1}{p_i^2}x_{-2}+p_i^3x_3
+\frac{1}{p_i^3}x_{-3}\cdots +\xi_{i0}$ and 
$\eta_i=q_ix_1+\frac{1}{q_i}x_{-1}
-q_i^2x_2-\frac{1}{q_i^2}x_{-2}+q_i^3x_3
+\frac{1}{q_i^3}x_{-3}\cdots +\eta_{i0}$.
\end{lemma}
\begin{proof}
See \cite{HirotaBook}.
\end{proof}

We impose the $C_{\infty}$-reduction $\theta_{n}=\theta_{-n}$ ($n\geq
0$) to the $A_{\infty}$ 2D-Toda system, {\it i.e.}, 
fold the infinite sequence $\{\dots,
 \theta_{-2},\theta_{-1},\theta_0,\theta_1,\theta_2,\dots\}$ 
in 
%%$\theta_{-1}$
$\theta_{0}$~\cite{Jimbo-Miwa,Ueno-Takasaki,Nimmo-Willox,Willox,DJKM2,DJKM3}. 
Then we have $\theta_{-1}=\theta_1$,
$\theta_{-2}=\theta_2$, $\theta_{-3}=\theta_3$, \dots . 

For $n=0$, 
\begin{eqnarray*}
\frac{\partial^2\theta_{0}}{\partial x_1 \partial x_{-1}}
&=&e^{-\theta_{-1}}-2e^{-\theta_{0}}+e^{-\theta_{1}}\\
&=&  -2e^{-\theta_{0}}+ 2e^{-\theta_{1}}\,.
\end{eqnarray*}
For $n=1$, 
\begin{equation*}
\frac{\partial^2\theta_1}{\partial x_1 \partial x_{-1}}
=e^{-\theta_{0}}-2e^{-\theta_{1}}+e^{-\theta_{2}}\,.
\end{equation*}
Thus we obtain
the $C_{\infty}$ 2D-Toda system~\cite{Leznov-Saveliev,Mikhailov3}:
\begin{equation}
\frac{\partial^2\theta_n}{\partial x_1 \partial x_{-1}}
=-\sum_{m\in \mathbf{Z}_{\geq 0}}a_{n+1\,,\,m+1}e^{-\theta_m}\,, 
\quad n\in \mathbf{Z}_{\geq 0}\,,\label{2dToda-C}
\end{equation}
where the matrix $A=(a_{n,m})$ is the transpose of the Cartan matrix for
the infinite dimensional Lie algebra $C_{\infty}$, 
which is the semi-infinite tridiagonal matrix
\begin{equation}
A=
\left[
\begin{array}{cccccccc}
2 & -2 &    & & & &  \\
-1 & 2 & -1 & & & &  \\
   & -1 & 2 & -1 & & &  \\
   &    & \ddots & \ddots & \ddots &  & \\
       &     &    &        -1    &  2    & -1 & \\
      &      &    &         &        \ddots & \ddots & \ddots     
\end{array}
\right]\,.\label{cartan-C}
\end{equation}

The $C_{\infty}$ 2D-Toda system (\ref{2dToda-C}) with 
(\ref{cartan-C})
may be written as 
\begin{eqnarray}
\frac{\partial^2\theta_{0}}{\partial x_1 \partial x_{-1}}
&=&-2e^{-\theta_{0}}+2e^{-\theta_{1}}\,,
\\
\frac{\partial^2\theta_n}{\partial x_1 \partial x_{-1}}
&=&e^{-\theta_{n-1}}-2e^{-\theta_{n}}+e^{-\theta_{n+1}}\,, \qquad n\geq 1\,.   
\end{eqnarray}
The $C_{\infty}$ 2D-Toda system (\ref{2dToda-C}) with 
(\ref{cartan-C}) is transformed into the bilinear equations 
\begin{eqnarray}
&&-\left(\frac{1}{2}D_{x_1}D_{x_{-1}}-1\right)\tau_{0}\cdot
 \tau_{0}=\tau_{1}^2\,,\label{CToda-bilinear1}\\
&&-\left(\frac{1}{2}D_{x_1}D_{x_{-1}}-1\right)\tau_n\cdot
 \tau_n=\tau_{n-1}\tau_{n+1}
\,, \quad {\rm for} \quad n\geq 1\,,\label{CToda-bilinear2}
\end{eqnarray}
through the dependent variable transformation
\begin{equation*}
\theta_{0}=-\ln \frac{\tau_1^2}{\tau_{0}^2}\,,\qquad {\rm and} 
\qquad 
\theta_n=-\ln \frac{\tau_{n+1}\tau_{n-1}}{\tau_n^2}\quad {\rm
for}\,\,\,\,n\geq 1.
\end{equation*}

\begin{lemma}\label{lemma:C-tau}
The bilinear equations of the $C_{\infty}$ 2D-Toda system
(\ref{CToda-bilinear1}) and (\ref{CToda-bilinear2})
have the $N$-soliton solution which is expressed as 
\begin{equation*}
\tau_n={\rm det}\left(\psi_{i,j}^{(n)}\right)_{1\leq i,j\leq 2N}\,,
\end{equation*}
where
\begin{eqnarray*}
&&\psi_{i,j}^{(n)}=c_{i,j}+(-1)^n\int_{-\infty}^{x_1}\varphi_i^{(n)}
\varphi_j^{(-n)}dx_1\,,\\
&&\varphi_i^{(n)}=p_i^ne^{\xi_i}\,,
\qquad \xi_i=p_ix_1+\frac{1}{p_i}x_{-1}+p_i^3x_3
+\frac{1}{p_i^3}x_{-3}+\cdots +\xi_{i0}\,,
\end{eqnarray*}
and $c_{i,j}=c_{j,i}$. 
\end{lemma}
\begin{proof}
Imposing the $C_{\infty}$ reduction $\tau_n=\tau_{-n}$, {\it i.e.}, 
folding the 
sequence of the $\tau$-functions $\{\dots,
 \tau_{-2},\tau_{-1},\tau_0,\tau_1,\tau_2,\dots\}$ 
in 
$\tau_{0}$, we have 
$\tau_{-1}=\tau_1$, $\tau_{-2}=\tau_2$, $\tau_{-3}=\tau_3$, ....
~\cite{Jimbo-Miwa,Ueno-Takasaki,Nimmo-Willox,Willox,DJKM2,DJKM3}.
Thus we obtain the bilinear equations (\ref{CToda-bilinear1}) 
and (\ref{CToda-bilinear2}) from the 2D-Toda bilinear equation
 (\ref{2dtl-bilinear}).  

To impose the $C_{\infty}$ reduction to the Gram-type determinant
 solution of the $A_{\infty}$ 2D-Toda system, 
we impose the constraint 
$\hat{\varphi}_j^{(n)}=\varphi_j^{(n)}$, $c_{i,j}=c_{j,i}$, 
$M=2N$ and $x_{2k}\equiv 0$ for every interger $k$. 
With this constraint, 
each element of the Gram-type determinant has the
following property: 
\begin{eqnarray*}
\psi_{i,j}^{(n)}&=&c_{i,j}+(-1)^n\int_{-\infty}^{x_1}\varphi_i^{(n)}
\varphi_j^{(-n)}dx_1\\
&=&c_{j,i}+(-1)^{-n}\int_{-\infty}^{x_1}\varphi_j^{(-n)}
\varphi_i^{(n)}dx_1\\
&=&\psi_{j,i}^{(-n)}\,.
\end{eqnarray*}
Then the $\tau$-function satisfies
$\tau_n=\tau_{-n}$.
Therefore the $N$-soliton solution of the $C_{\infty}$ 2D-Toda system 
is expressed by the above Gram-type determinant. 
\end{proof}

\begin{lemma}
The bilinear equations
\begin{eqnarray}
&&-\left(\frac{1}{2}D_{x_1}D_{x_{-1}}-1\right)\tau_{0}\cdot
 \tau_{0}=\tau_1^2\,,\label{3CToda-bilinear1}\\
&&-\left(\frac{1}{2}D_{x_1}D_{x_{-1}}-1\right)\tau_{1} \cdot
 \tau_{1}=\tau_{0}\tau_{2}\,,
\label{3CToda-bilinear2}
\end{eqnarray}
are satisfied by the $\tau$-functions which are obtained by the 
pseudo 3-reduction of the $C_{\infty}$ 2D-Toda system. 
The $\tau$-functions are given by 
\begin{equation*}
\tau_n={\rm det}\left(
\psi_{i,j}^{(n)}\right)_{1\leq i,j \leq 2N}\,,
\end{equation*}
where
\begin{eqnarray*}
&&\psi_{i,j}^{(n)}=c_{i,j}+(-1)^n\int_{-\infty}^{x_1}\varphi_i^{(n)}
\varphi_j^{(-n)}dx_1\,,\\
&&\varphi_i^{(n)}=p_i^ne^{\xi_i}\,,
\quad \xi_i=p_ix_1+\frac{1}{p_i}x_{-1}+p_i^3x_3
+\frac{1}{p_i^3}x_{-3}+\cdots +\xi_{i0}\,,
\end{eqnarray*}
and $c_{i,j}=\delta_{j,2N+1-i}\alpha_i$, $\alpha_i=\alpha_{2N+1-i}$, 
 $p_i^2-p_ip_{2N+1-i}+p_{2N+1-i}^2=1$.
\end{lemma}
\begin{proof}
To impose the pseudo 3-reduction to the $\tau$-function in Lemma \ref{lemma:C-tau}, 
we add a constraint $p_i^3+p_{2N+1-i}^3=p_i+p_{2N+1-i}$, $p_i\neq
 -p_{2N+1-i}$, i.e. $p_i^2-p_ip_{2N+1-i}+p_{2N+1-i}^2=1$,  
and $c_{i,j}=\delta_{j,2N+1-i}\alpha_i$, $\alpha_i=\alpha_{2N+1-i}$~\cite{Hirota-reduction}. 
\end{proof}

\begin{lemma}\label{lemma:pfaffian-C}
The $\tau$-function 
$\tau_1$ of the bilinear equations (\ref{3CToda-bilinear1}) and 
(\ref{3CToda-bilinear2}) with the pseudo 3-reduction constraint 
satisfies the following relation:
\begin{equation}
\tau_1=\frac{1}{c}
\,g_1g_2\,,\label{tau1-g1-g2}
\end{equation}
with
\begin{eqnarray}
&&g_1= {\rm pf}\left(2\alpha_i\frac{(p_i-1)(p_j-1)}{p_i-p_j}\delta_{j,2N+1-i}
+\frac{p_i-p_j}{p_i+p_j}e^{\xi_i+\xi_j}\right)_{1\leq i,j\leq
2N}\,,\label{g1}\\
&&g_2={\rm pf}\left(2\alpha_i\frac{(p_i+1)(p_j+1)}{p_i-p_j}\delta_{j,2N+1-i}
+\frac{p_i-p_j}{p_i+p_j}e^{\xi_i+\xi_j}\right)_{1\leq
i,j\leq 2N}\,,\label{g2}
\end{eqnarray}
 where $\xi_i=p_i^{-1}x_{-1}+p_ix_1+\xi_i^0$, 
$p_i^2-p_ip_{2N+1-i}+p_{2N+1-i}^2=1$\,, $\alpha_i=\alpha_{2N+1-i}$,
 $c=2^{2N}\prod_{k=1}^{2N}p_k$.   
\end{lemma}
\begin{proof}
Suppose that $\alpha_i=\alpha_{2N+1-i}$ and
$p_i^2-p_ip_{2N+1-i}+p_{2N+1-i}^2=1$ are satisfied. 
Then we can rewrite $\tau_1$ as follows:
\begin{eqnarray*}
\fl \tau_1&=&
{\rm det}\left(\psi_{i,j}^{(1)}\right)_{1\leq i,j\leq 2N}
={\rm det}\left(\delta_{j,2N+1-i}\alpha_i+\frac{1}{p_i+p_j}\left(
-\frac{p_i}{p_j}\right)e^{\xi_i+\xi_j}\right)_{1\leq i,j\leq 2N}\\
\fl &=&{\rm
 det}\left(\delta_{j,2N+1-i}\,\alpha_i\frac{p_ip_{2N+1-i}-p_{2N+1-i}^2}{p_j(p_i-p_j)}
-\frac{1}{p_i+p_j}
\frac{p_i}{p_j}e^{\xi_i+\xi_j}\right)_{1\leq i,j\leq 2N}\\
\fl &=&{\rm
 det}\left(\delta_{j,2N+1-i}\,\alpha_i\frac{p_i^2-1}{p_j(p_i-p_j)}
-\frac{1}{p_i+p_j}
\frac{p_i}{p_j}e^{\xi_i+\xi_j}\right)_{1\leq i,j\leq 2N}\\
\fl &=&\frac{1}{c}\,{\rm
 det}\left(2\delta_{j,2N+1-i}\,\alpha_i\frac{p_i^2-1}{p_i-p_j}
-\frac{2p_i}{p_i+p_j}e^{\xi_i+\xi_j}\right)_{1\leq
 i,j\leq 2N}
\\
\fl &=&\frac{1}{c}\,{\rm
 det}\left(2\delta_{j,2N+1-i}\,\alpha_i\frac{p_i^2-1}{p_i-p_j}
+\left(\frac{(p_i-p_j)(p_i+1)}{(p_i+p_j)(p_j-1)}-\frac{p_i-1}{p_j-1}\right)
e^{\xi_i+\xi_j}\right)_{1\leq
 i,j\leq 2N}
\\
\fl &=&\frac{1}{c}\,{\rm
 det}\left(2\delta_{j,2N+1-i}\,\alpha_i\frac{p_i^2-1}{p_i-p_j}
+\frac{(p_i-p_j)(p_i+1)}{(p_i+p_j)(p_j-1)}e^{\xi_i+\xi_j}-\frac{p_i-1}{p_j-1}e^{\xi_i+\xi_j}
\right)_{1\leq
 i,j\leq 2N}
\\
%%%% &=\frac{\prod_{k=1}^{2N}\frac{p_k+1}{p_k-1}}{c}\,{\rm
%%%% det}\left(2\delta_{j,2N+1-i}\,\alpha_i\frac{(p_i-1)(p_j-1)}{p_i-p_j}
%%%%+\frac{p_i-p_j}{p_i+p_j}e^{\xi_i+\xi_j}-\frac{p_i-1}{p_i+1}e^{\xi_i+\xi_j}
%%%%\right)_{1\leq i,j\leq 2N}\\
\fl &=&\frac{\prod_{k=1}^{2N}\frac{p_k+1}{p_k-1}}
{c}\,{\rm
 det}\left(c_{i,j}
+\frac{p_i-p_j}{p_i+p_j}e^{\xi_i+\xi_j}-\frac{p_i-1}{p_i+1}e^{\xi_i+\xi_j}
\right)_{1\leq
 i,j\leq 2N}\,,
\end{eqnarray*}
where
 $c_{i,j}=2\delta_{j,2N+1-i}\,\alpha_i\frac{(p_i-1)(p_j-1)}{p_i-p_j}$
 and $c=2^{2N}\prod_{k=1}^{2N}p_k$. 
Using the formula (\ref{formula1}), we obtain
\begin{equation*}
 \tau_1=\frac{\prod_{k=1}^{2N}\frac{p_k+1}{p_k-1}}
{c}\,
\left|
\begin{array}{ccccc}
\Psi_{1,1} & \Psi_{1,2} & \cdots  & \Psi_{1,2N} & \frac{p_1-1}{p_1+1}e^{\xi_1} \\
\vdots & \vdots  & \ddots  & \vdots &\vdots  \\
\Psi_{2N,1} & \Psi_{2N,2}  & \cdots  & \Psi_{2N,2N} & \frac{p_{2N}-1}{p_{2N}+1}e^{\xi_{2N}}  \\
e^{\xi_1} & e^{\xi_2} & \cdots  & e^{\xi_{2N}} & 1 
\end{array}
\right|\,,\quad 
\end{equation*}
where $\Psi_{i,j}=c_{i,j}
+\frac{p_i-p_j}{p_i+p_j}e^{\xi_i+\xi_j}$. 
Then we note
\[
\fl 
c_{i,j}=2\delta_{j,2N+1-i}\,\alpha_i\frac{(p_i-1)(p_j-1)}{p_i-p_j}=
-2\delta_{i,2N+1-j}\,\alpha_{2N+1-i}\frac{(p_i-1)(p_j-1)}{p_j-p_i}
=-c_{j,i}\,. 
\]
Introducing $c_i=2\alpha_i\frac{(p_i-1)(p_j-1)}{p_i-p_j}$, 
we can write as $c_{i,j}=\delta_{j,2N+1-i}c_i$ and $c_i=-c_{2N+1-i}$.
Note that the $2N \times 2N$ matrix $(\Psi_{i,j})_{1\leq i,j\leq 2N}$ is skew-symmetric. 
Thus we can use the formula (\ref{formula3}): 
\begin{eqnarray*}
\fl \tau_1&=&\frac{\prod_{k=1}^{2N}\frac{p_k+1}{p_k-1}}
{c}\,
{\rm pf}(\Psi_{i,j})_{1\leq i,j\leq 2N}\,
{\rm pf}\left(\Psi_{ij}-\frac{p_i-1}{p_i+1}e^{\xi_i+\xi_j}+
\frac{p_j-1}{p_j+1}e^{\xi_i+\xi_j}\right)_{1\leq i,j\leq 2N}
\nonumber\\
\fl &=&\frac{\prod_{k=1}^{2N}\frac{p_k+1}{p_k-1}}
{c}\,
{\rm pf}\left(2\alpha_i\frac{(p_i-1)(p_j-1)}{p_i-p_j}\delta_{j,2N+1-i}
+\frac{p_i-p_j}{p_i+p_j}e^{\xi_i+\xi_j}\right)_{1\leq i,j\leq
2N}\,\nonumber\\
\fl & \quad& \times
{\rm pf}\left(2\alpha_i\frac{(p_i-1)(p_j-1)}{p_i-p_j}\delta_{j,2N+1-i}
+\frac{(p_i-p_j)(p_i-1)(p_j-1)}{(p_i+p_j)(p_i+1)(p_j+1)}e^{\xi_i+\xi_j}\right)_{1\leq
i,j\leq 2N}\nonumber\\
\fl &=&\frac{1}
{c}\,
{\rm pf}\left(2\alpha_i\frac{(p_i-1)(p_j-1)}{p_i-p_j}\delta_{j,2N+1-i}
+\frac{p_i-p_j}{p_i+p_j}e^{\xi_i+\xi_j}\right)_{1\leq i,j\leq
2N}\,\nonumber\\
\fl & \quad &  \times
{\rm pf}\left(2\alpha_i\frac{(p_i+1)(p_j+1)}{p_i-p_j}\delta_{j,2N+1-i}
+\frac{p_i-p_j}{p_i+p_j}e^{\xi_i+\xi_j}\right)_{1\leq
i,j\leq 2N}\nonumber\\
\fl & =& \frac{1}{c}g_1g_2\,,
\end{eqnarray*}
where $c=2^{2N}\prod_{k=1}^{2N}p_k$. 
\end{proof}

\begin{lemma}\label{lemma:pfaffian-C-2}
The $\tau$-function 
$\tau_0$ of the bilinear equations (\ref{3CToda-bilinear1}) and 
(\ref{3CToda-bilinear2}) with the pseudo 3-reduction constraint 
satisfies the relation
\begin{equation}
\tau_0=\frac{1}{c}
\,(g_1g_2-D_{x_1}g_1\cdot g_2)\,,\label{tau0-g1-g2}
\end{equation}
with
\begin{eqnarray*}
&&g_1= {\rm pf}\left(2\alpha_i\frac{(p_i-1)(p_j-1)}{p_i-p_j}\delta_{j,2N+1-i}
+\frac{p_i-p_j}{p_i+p_j}e^{\xi_i+\xi_j}\right)_{1\leq i,j\leq
2N}\,,\\
&&g_2={\rm pf}\left(2\alpha_i\frac{(p_i+1)(p_j+1)}{p_i-p_j}\delta_{j,2N+1-i}
+\frac{p_i-p_j}{p_i+p_j}e^{\xi_i+\xi_j}\right)_{1\leq
i,j\leq 2N}\,,
\end{eqnarray*}
 where $\xi_i=p_i^{-1}x_{-1}+p_ix_1+\xi_i^0$, 
$p_i^2-p_ip_{2N+1-i}+p_{2N+1-i}^2=1$\,, $\alpha_i=\alpha_{2N+1-i}$,
 $c=2^{2N}\prod_{k=1}^{2N}p_k$.   
\end{lemma}
\begin{proof}
Suppose that $\alpha_i=\alpha_{2N+1-i}$ and
$p_i^2-p_ip_{2N+1-i}+p_{2N+1-i}^2=1$ are satisfied. 
Then we can rewrite $\tau_0$ as follows:
\begin{eqnarray*}
\fl  \tau_0&=&
{\rm det}\left(\psi_{i,j}^{(0)}\right)_{1\leq i,j\leq 2N}\
={\rm det}\left(\delta_{j,2N+1-i}\alpha_i+\frac{1}{p_i+p_j}
e^{\xi_i+\xi_j}\right)_{1\leq i,j\leq 2N}\\
\fl  &=&{\rm
 det}\left(\delta_{j,2N+1-i}\,\alpha_i\frac{p_ip_{2N+1-i}-p_{2N+1-i}^2}{p_j(p_i-p_j)}
+\frac{1}{p_i+p_j}
e^{\xi_i+\xi_j}\right)_{1\leq i,j\leq 2N}\\
\fl  &=&{\rm
 det}\left(\delta_{j,2N+1-i}\,\alpha_i\frac{p_i^2-1}{p_j(p_i-p_j)}
+\frac{1}{p_i+p_j}
e^{\xi_i+\xi_j}\right)_{1\leq i,j\leq 2N}\\
\fl  &=&\frac{1}{c}\,{\rm
 det}\left(2\delta_{j,2N+1-i}\,\alpha_i\frac{p_i^2-1}{p_i-p_j}
+\frac{2p_j}{p_i+p_j}e^{\xi_i+\xi_j}\right)_{1\leq
 i,j\leq 2N}
\\
\fl &=&\frac{1}{c}\,{\rm
 det}\left(2\delta_{j,2N+1-i}\,\alpha_i\frac{p_i^2-1}{p_i-p_j}
%%%\right.\\
%%&\quad& \qquad \left.
+\left(\frac{(p_i-p_j)(p_i+1)}{(p_i+p_j)(p_j-1)}-\frac{p_i+1}{p_j-1}+\frac{2p_j}{p_j-1}\right)
e^{\xi_i+\xi_j}\right)_{1\leq
 i,j\leq 2N}
\,,
\end{eqnarray*}
where $c=2^{2N}\prod_{k=1}^{2N}p_k$. 
Using the formula (\ref{formula1}), we can rewrite $\tau_0$ as follows:
\begin{eqnarray*}
\tau_0&=&\frac{1}
{c}\,
\left|
\begin{array}{cccccc}
\Phi_{1,1} & \Phi_{1,2} & \cdots & \Phi_{1,2N} &
-2e^{\xi_1} & (p_1+1)e^{\xi_1} \\
\vdots & \vdots  & \ddots & \vdots &\vdots &\vdots  \\
\Phi_{2N,1} & \Phi_{2N,2}  & \cdots & \Phi_{2N,2N} & -2e^{\xi_{2N}} 
& (p_{2N}+1)e^{\xi_{2N}}  \\
\frac{p_1}{p_{1}-1}e^{\xi_1} & \frac{p_2}{p_{2}-1}e^{\xi_2} & \cdots 
& \frac{p_{2N}}{p_{2N}-1}e^{\xi_{2N}} & 1 &0\\
\frac{1}{p_{1}-1}e^{\xi_1} & \frac{1}{p_{2}-1}e^{\xi_2} & \cdots  & 
\frac{1}{p_{2N}-1}e^{\xi_{2N}} & 0 &1
\end{array}
\right|\\
 &=&
\frac{1}
{c}\,
\left|
\begin{array}{cccccc}
\Phi_{1,1} & \Phi_{1,2} & \cdots & \Phi_{1,2N} &
(p_1-1)e^{\xi_1} & (p_1+1)e^{\xi_1} \\
\vdots & \vdots  & \ddots & \vdots  & \vdots  & \vdots \\
\Phi_{2N,1} & \Phi_{2N,2}  & \cdots  & \Phi_{2N,2N} & (p_{2N}-1)e^{\xi_{2N}} 
& (p_{2N}+1)e^{\xi_{2N}}  \\
\frac{p_1}{p_{1}-1}e^{\xi_1} & \frac{p_2}{p_{2}-1}e^{\xi_2} & \cdots
& \frac{p_{2N}}{p_{2N}-1}e^{\xi_{2N}} & 1 &0\\
\frac{1}{p_{1}-1}e^{\xi_1} & \frac{1}{p_{2}-1}e^{\xi_2} & \cdots  &
\frac{1}{p_{2N}-1}e^{\xi_{2N}} & 1 &1
\end{array}
\right|\,,
\end{eqnarray*}
where $\Phi_{i,j}=2\delta_{j,2N+1-i}\,\alpha_i\frac{p_i^2-1}{p_i-p_j}
+\frac{(p_i-p_j)(p_i+1)}{(p_i+p_j)(p_j-1)}e^{\xi_i+\xi_j}$. 
Then we can further simplify as follows:
\begin{equation*}
\tau_0 =
\frac{\prod_{k=1}^{2N}\frac{p_k+1}{p_k-1}}
{c}\,\left|
\begin{array}{cccccc}
\Psi_{1,1} & \Psi_{1,2} & \cdots & \Psi_{1,2N} &
\frac{p_1-1}{p_1+1}e^{\xi_1} & e^{\xi_1} \\
 \vdots & \vdots  & \ddots &\vdots &\vdots &\vdots  \\
\Psi_{2N,1} & \Psi_{2N,2}  & \cdots  & \Psi_{2N,2N} & \frac{p_{2N}-1}{p_{2N}+1}e^{\xi_{2N}} 
& e^{\xi_{2N}}  \\
p_1e^{\xi_1} & p_2e^{\xi_2} & \cdots 
& p_{2N}e^{\xi_{2N}} & 1 &0\\
e^{\xi_1} & e^{\xi_2} & \cdots & 
e^{\xi_{2N}} & 1 &1
\end{array}
\right|
\,,
\end{equation*}
where $\Psi_{i,j}=2\delta_{j,2N+1-i}\,\alpha_i\frac{(p_i-1)(p_j-1)}{p_i-p_j}
+\frac{p_i-p_j}{p_i+p_j}e^{\xi_i+\xi_j}$. 
Note that the $2N \times 2N$ matrix $(\Psi_{i,j})_{1\leq i,j \leq 2N}$ is skew-symmetric. 
Thus we can use the formula (\ref{formula4}): 
\begin{eqnarray*}
\tau_0
&=&\frac{\prod_{k=1}^{2N}\frac{p_k+1}{p_k-1}}
{c}\,\left(
\left.
\begin{array}{cccccccc}
| & \Psi_{1,2} & \Psi_{1,3} &\cdots & 
 \Psi_{1,2N} & e^{\xi_1} & e^{\xi_1}\\
&  & \Psi_{2,3} & \cdots  & \Psi_{2,2N} & e^{\xi_2} & e^{\xi_2}\\
&  &  & \ddots  &\vdots & \vdots &\vdots \\
 &  & &  &   \Psi_{2N-1, 2N} & e^{\xi_{2N-1}} & e^{\xi_{2N-1}} \\
 &  & &  &        & e^{\xi_{2N}} & e^{\xi_{2N}} \\
&  &  & & &    & 1
\end{array}
\right|\right.\\
 && \times \left.
\begin{array}{cccccccc}
| & \Psi_{1,2} & \Psi_{1,3} &\cdots & 
 \Psi_{1,2N} & \frac{p_1-1}{p_1+1}e^{\xi_1} & p_1e^{\xi_1}\\
&  & \Psi_{2,3} & \cdots  & \Psi_{2,2N} & \frac{p_2-1}{p_2+1}e^{\xi_2} & p_2e^{\xi_2}\\
&  &  & \ddots  &\vdots & \vdots &\vdots \\
 &  & &  &   \Psi_{2N-1, 2N} & \frac{p_{2N-1}-1}{p_{2N-1}+1}e^{\xi_{2N-1}} & p_{2N-1}e^{\xi_{2N-1}} \\
 &  & &  &        & \frac{p_{2N}-1}{p_{2N}+1}e^{\xi_{2N}} & p_{2N}e^{\xi_{2N}} \\
&  &  & & &    & 1
\end{array}
\right|\\
 && -
\left.
\begin{array}{cccccccc}
| & \Psi_{1,2} & \Psi_{1,3} &\cdots & 
 \Psi_{1,2N} & e^{\xi_1} & p_1e^{\xi_1}\\
&  & \Psi_{2,3} & \cdots  & \Psi_{2,2N} & e^{\xi_2} & p_2e^{\xi_2}\\
&  &  & \ddots  &\vdots & \vdots &\vdots \\
 &  & &  &   \Psi_{2N-1, 2N} & e^{\xi_{2N-1}} & p_{2N-1}e^{\xi_{2N-1}} \\
 &  & &  &        & e^{\xi_{2N}} & p_{2N}e^{\xi_{2N}} \\
&  &  & & &    & 0
\end{array}
\right|\\
 && \left. \times \left.
\begin{array}{cccccccc}
| & \Psi_{1,2} & \Psi_{1,3} &\cdots & 
 \Psi_{1,2N} & \frac{p_1-1}{p_1+1}e^{\xi_1} & e^{\xi_1}\\
&  & \Psi_{2,3} & \cdots  & \Psi_{2,2N} & \frac{p_2-1}{p_2+1}e^{\xi_2} & e^{\xi_2}\\
&  &  & \ddots  &\vdots & \vdots &\vdots \\
 &  & &  &   \Psi_{2N-1, 2N} & \frac{p_{2N-1}-1}{p_{2N-1}+1}e^{\xi_{2N-1}} & e^{\xi_{2N-1}} \\
 &  & &  &        & \frac{p_{2N}-1}{p_{2N}+1}e^{\xi_{2N}} & e^{\xi_{2N}} \\
&  &  & & &    & 1
\end{array}
\right|\right)
%%%\\ &=\frac{\prod_{k=1}^{2N}\frac{p_k+1}{p_k-1}}
%%%{c}\,\left[
%%%{\rm pf}(\Psi_{i,j})_{1\leq i,j\leq 2N}\,
%%%{\rm pf}\left(\Psi_{ij}-\frac{p_i-1}{p_i+1}p_je^{\xi_i+\xi_j}+
%%%p_i\frac{p_j-1}{p_j+1}e^{\xi_i+\xi_j}\right)_{1\leq i,j\leq 2N}\right.\nonumber\\
%%% & -\left.
%%%\begin{array}{cccccccc}
%%%| & \Psi_{1,2} & \Psi_{1,3} &\cdots & 
%%% \Psi_{1,2N} & e^{\xi_1} & p_1e^{\xi_1}\\
%%%&  & \Psi_{2,3} & \cdots  & \Psi_{2,2N} & e^{\xi_2} & p_2e^{\xi_2}\\
%%%&  &  & \ddots  &\vdots & \vdots &\vdots \\
%%%&  & &  &   \Psi_{2N-1, 2N} & e^{\xi_{2N-1}} & p_{2N-1}e^{\xi_{2N-1}} \\
%%% &  & &  &        & e^{\xi_{2N}} & p_{2N}e^{\xi_{2N}} \\
%%%&  &  & & &    & 0
%%%\end{array}
%%%\right|\\
%%% & \left.\times {\rm pf}\left(\Psi_{ij}-\frac{p_i-1}{p_i+1}e^{\xi_i+\xi_j}+
%%%\frac{p_j-1}{p_j+1}e^{\xi_i+\xi_j}\right)_{1\leq i,j\leq 2N}\right]\,.
\end{eqnarray*}
Let
\begin{equation*}
\fl  g_1= {\rm pf}\left(\Psi_{i,j}\right)_{1\leq i,j\leq
2N}\\
  =\left.
\begin{array}{cccccccc}
| & \Psi_{1,2} & \Psi_{1,3} &\cdots & 
 \Psi_{1,2N} & e^{\xi_1} & e^{\xi_1}\\
&  & \Psi_{2,3} & \cdots  & \Psi_{2,2N} & e^{\xi_2} & e^{\xi_2}\\
&  &  & \ddots  &\vdots & \vdots &\vdots \\
 &  & &  &   \Psi_{2N-1, 2N} & e^{\xi_{2N-1}} & e^{\xi_{2N-1}} \\
 &  & &  &        & e^{\xi_{2N}} & e^{\xi_{2N}} \\
&  &  & & &    & 1
\end{array}
\right|\,,\\
\end{equation*}
\begin{eqnarray*}
\fl  g_2&=&{\rm pf}\left(2\alpha_i\frac{(p_i+1)(p_j+1)}{p_i-p_j}\delta_{j,2N+1-i}
+\frac{p_i-p_j}{p_i+p_j}e^{\xi_i+\xi_j}\right)_{1\leq
i,j\leq 2N}\\
\fl  &=&\prod_{k=1}^{2N}\frac{p_k+1}{p_k-1}\,{\rm pf}\left(\Psi_{ij}-\frac{p_i-1}{p_i+1}e^{\xi_i+\xi_j}+
\frac{p_j-1}{p_j+1}e^{\xi_i+\xi_j}\right)_{1\leq i,j\leq 2N}\\
\fl  &=&\prod_{k=1}^{2N}\frac{p_k+1}{p_k-1}\times
\left.
\begin{array}{cccccccc}
| & \Psi_{1,2} & \Psi_{1,3} &\cdots & 
 \Psi_{1,2N} & \frac{p_1-1}{p_1+1}e^{\xi_1} & e^{\xi_1}\\
&  & \Psi_{2,3} & \cdots  & \Psi_{2,2N} & \frac{p_2-1}{p_2+1}e^{\xi_2} & e^{\xi_2}\\
&  &  & \ddots  &\vdots & \vdots &\vdots \\
 &  & &  &   \Psi_{2N-1, 2N} & \frac{p_{2N-1}-1}{p_{2N-1}+1}e^{\xi_{2N-1}} & e^{\xi_{2N-1}} \\
 &  & &  &        & \frac{p_{2N}-1}{p_{2N}+1}e^{\xi_{2N}} & e^{\xi_{2N}} \\
&  &  & & &    & 1
\end{array}
\right|\,.
\end{eqnarray*}
Then
\begin{eqnarray*}
\fl  &&\partial_{x_1}g_1=\left.
\begin{array}{cccccccc}
| & \Psi_{1,2} & \Psi_{1,3} &\cdots & 
 \Psi_{1,2N} & e^{\xi_1} & p_1e^{\xi_1}\\
&  & \Psi_{2,3} & \cdots  & \Psi_{2,2N} & e^{\xi_2} & p_2e^{\xi_2}\\
&  &  & \ddots  &\vdots & \vdots &\vdots \\
 &  & &  &   \Psi_{2N-1, 2N} & e^{\xi_{2N-1}} & p_{2N-1}e^{\xi_{2N-1}} \\
 &  & &  &        & e^{\xi_{2N}} & p_{2N}e^{\xi_{2N}} \\
&  &  & & &    & 0
\end{array}
\right|\,,\\
\fl  &&(\partial_{x_1}+1)g_2=
(\partial_{x_1}+1){\rm pf}\left(2\alpha_i\frac{(p_i+1)(p_j+1)}{p_i-p_j}\delta_{j,2N+1-i}
+\frac{p_i-p_j}{p_i+p_j}e^{\xi_i+\xi_j}\right)_{1\leq
i,j\leq 2N}
\\
\fl  &&\qquad  = 
{\rm pf}\left(2\alpha_i\frac{(p_i+1)(p_j+1)}{p_i-p_j}\delta_{j,2N+1-i}
+\frac{p_i-p_j}{p_i+p_j}e^{\xi_i+\xi_j}+(p_i-p_j)e^{\xi_i+\xi_j}\right)_{1\leq
i,j\leq 2N}
\\
\fl  &&\qquad  = \prod_{k=1}^{2N}\frac{p_k+1}{p_k-1}\,
{\rm pf}\left(\Psi_{ij}-\frac{p_i-1}{p_i+1}p_je^{\xi_i+\xi_j}+
p_i\frac{p_j-1}{p_j+1}e^{\xi_i+\xi_j}\right)_{1\leq i,j\leq 2N}
\\
\fl  &&\qquad =\prod_{k=1}^{2N}\frac{p_k+1}{p_k-1}\times 
\left.
\begin{array}{cccccccc}
| & \Psi_{1,2} & \Psi_{1,3} &\cdots & 
 \Psi_{1,2N} & \frac{p_1-1}{p_1+1}e^{\xi_1} & p_1e^{\xi_1}\\
&  & \Psi_{2,3} & \cdots  & \Psi_{2,2N} & \frac{p_2-1}{p_2+1}e^{\xi_2} & p_2e^{\xi_2}\\
&  &  & \ddots  &\vdots & \vdots &\vdots \\
 &  & &  &   \Psi_{2N-1, 2N} & \frac{p_{2N-1}-1}{p_{2N-1}+1}e^{\xi_{2N-1}} & p_{2N-1}e^{\xi_{2N-1}} \\
 &  & &  &        & \frac{p_{2N}-1}{p_{2N}+1}e^{\xi_{2N}} & p_{2N}e^{\xi_{2N}} \\
&  &  & & &    & 1
\end{array}
\right|\,.
\end{eqnarray*}
Thus $\tau_0$, $g_1$ and $g_2$ satisfy the relation (\ref{tau0-g1-g2}). 
\end{proof}

\begin{lemma}\label{lemma:pfaffian-C-3}
The $\tau$-function 
$\tau_2$ of the bilinear equations (\ref{3CToda-bilinear1}) and 
(\ref{3CToda-bilinear2}) with the pseudo 3-reduction constraint 
satisfies the relation
\begin{equation}
\tau_2=\frac{1}{c}
\,(g_1g_2-D_{x_{-1}}g_1\cdot g_2)\,,\label{tau2-g1-g2}
\end{equation}
with
\begin{eqnarray*}
&&g_1= {\rm pf}\left(2\alpha_i\frac{(p_i-1)(p_j-1)}{p_i-p_j}\delta_{j,2N+1-i}
+\frac{p_i-p_j}{p_i+p_j}e^{\xi_i+\xi_j}\right)_{1\leq i,j\leq
2N}\,,\\
&&g_2={\rm pf}\left(2\alpha_i\frac{(p_i+1)(p_j+1)}{p_i-p_j}\delta_{j,2N+1-i}
+\frac{p_i-p_j}{p_i+p_j}e^{\xi_i+\xi_j}\right)_{1\leq
i,j\leq 2N}\,,
\end{eqnarray*}
 where $\xi_i=p_i^{-1}x_{-1}+p_ix_1+\xi_i^0$, 
$p_i^2-p_ip_{2N+1-i}+p_{2N+1-i}^2=1$\,, $\alpha_i=\alpha_{2N+1-i}$,
 $c=2^{2N}\prod_{k=1}^{2N}p_k$.   
\end{lemma}
\begin{proof}
Suppose that $\alpha_i=\alpha_{2N+1-i}$ and
$p_i^2-p_ip_{2N+1-i}+p_{2N+1-i}^2=1$ are satisfied. 
Then we can rewrite $\tau_2$ as follows:
\begin{eqnarray*}
\fl  \tau_2&=&
{\rm det}\left(\psi_{i,j}^{(2)}\right)_{1\leq i,j\leq 2N}\\
\fl  &=&{\rm
 det}\left(\delta_{j,2N+1-i}\,\alpha_i\frac{p_{2N+1-i}(p_i-p_{2N+1-i})}{p_j(p_i-p_j)}
+\frac{1}{p_i+p_j}
\frac{p_i^2}{p_j^2}e^{\xi_i+\xi_j}\right)_{1\leq i,j\leq 2N}\\
\fl  &=&{\rm
 det}\left(\delta_{j,2N+1-i}\,\alpha_i\frac{p_ip_{2N+1-i}-p_{2N+1-i}^2}{p_j(p_i-p_j)}
+\frac{1}{p_i+p_j}
\frac{p_i^2}{p_j^2}e^{\xi_i+\xi_j}\right)_{1\leq i,j\leq 2N}\\
\fl  &=&{\rm
 det}\left(\delta_{j,2N+1-i}\,\alpha_i\frac{p_i^2-1}{p_j(p_i-p_j)}
+\frac{1}{p_i+p_j}
\frac{p_i^2}{p_j^2}e^{\xi_i+\xi_j}\right)_{1\leq i,j\leq 2N}\\
\fl  &=&\frac{1}{c}\,{\rm
 det}\left(2\delta_{j,2N+1-i}\,\alpha_i\frac{p_i^2-1}{p_i-p_j}
+\frac{2p_i^2}{p_i+p_j}\frac{1}{p_j}e^{\xi_i+\xi_j}\right)_{1\leq
 i,j\leq 2N}
\\
\fl  &=&\frac{1}{c}\,{\rm
 det}\left(2\delta_{j,2N+1-i}\,\alpha_i\frac{p_j}{p_i}\frac{p_i^2-1}{p_i-p_j}
+\frac{2p_i}{p_i+p_j}e^{\xi_i+\xi_j}\right)_{1\leq
 i,j\leq 2N}
\\
\fl  &=&\frac{1}{c}\,{\rm
 det}\left(2\delta_{j,2N+1-i}\,\alpha_i\frac{p_i^{-2}-1}
{p_i^{-1}-p_j^{-1}}
+\frac{2p_i}{p_i+p_j}e^{\xi_i+\xi_j}\right)_{1\leq
 i,j\leq 2N}
\\
\fl  &=&\frac{1}{c}\,{\rm
 det}\left(2\delta_{j,2N+1-i}\,\alpha_i\frac{p_i^{-2}-1}
{p_i^{-1}-p_j^{-1}}\right.\\
\fl &&\left.\qquad +\left(\frac{(p_i^{-1}-p_j^{-1})
(p_i^{-1}+1)}{(p_i^{-1}+p_j^{-1})(p_j^{-1}-1)}
-\frac{p_i^{-1}+1}{p_j^{-1}-1}+\frac{2p_j^{-1}}{p_j^{-1}-1}\right)
e^{\xi_i+\xi_j}\right)_{1\leq
 i,j\leq 2N}
\,,
\end{eqnarray*}
where $c=2^{2N}\prod_{k=1}^{2N}p_k$. 
Using the formula (\ref{formula1}), we can rewrite $\tau_2$ as follows:
\begin{eqnarray*}
\fl  &&  \tau_2=\frac{1}
{c}\,
\left|
\begin{array}{cccccc}
\tilde{\Phi}_{1,1} & \tilde{\Phi}_{1,2} & \cdots & \tilde{\Phi}_{1,2N} &
-2e^{\xi_1} & (p_1^{-1}+1)e^{\xi_1} \\
\vdots & \vdots  & \ddots & \vdots &\vdots &\vdots  \\
\tilde{\Phi}_{2N,1} & \tilde{\Phi}_{2N,2}  & \cdots & \tilde{\Phi}_{2N,2N} & -2e^{\xi_{2N}} 
& (p_{2N}^{-1}+1)e^{\xi_{2N}}  \\
\frac{p_1^{-1}}{p_{1}^{-1}-1}e^{\xi_1} & \frac{p_2^{-1}}{p_{2}^{-1}-1}e^{\xi_2} & \cdots 
& \frac{p_{2N}^{-1}}{p_{2N}^{-1}-1}e^{\xi_{2N}} & 1 &0\\
\frac{1}{p_{1}^{-1}-1}e^{\xi_1} & \frac{1}{p_{2}^{-1}-1}e^{\xi_2} & \cdots  & 
\frac{1}{p_{2N}^{-1}-1}e^{\xi_{2N}} & 0 &1
\end{array}
\right|\\
\fl  && \quad =
\frac{1}
{c}\,
\left|
\begin{array}{cccccc}
\tilde{\Phi}_{1,1} & \tilde{\Phi}_{1,2} & \cdots & \tilde{\Phi}_{1,2N} &
(p_1^{-1}-1)e^{\xi_1} & (p_1^{-1}+1)e^{\xi_1} \\
\vdots & \vdots  & \ddots & \vdots  & \vdots  & \vdots \\
\tilde{\Phi}_{2N,1} & \tilde{\Phi}_{2N,2}  & \cdots  & \tilde{\Phi}_{2N,2N} & (p_{2N}^{-1}-1)e^{\xi_{2N}} 
& (p_{2N}^{-1}+1)e^{\xi_{2N}}  \\
\frac{p_1^{-1}}{p_{1}^{-1}-1}e^{\xi_1} & \frac{p_2^{-1}}{p_{2}^{-1}-1}e^{\xi_2} & \cdots
& \frac{p_{2N}^{-1}}{p_{2N}^{-1}-1}e^{\xi_{2N}} & 1 &0\\
\frac{1}{p_{1}^{-1}-1}e^{\xi_1} & \frac{1}{p_{2}^{-1}-1}e^{\xi_2} & \cdots  &
\frac{1}{p_{2N}^{-1}-1}e^{\xi_{2N}} & 1 &1
\end{array}
\right|\,,
\end{eqnarray*}
where $\tilde{\Phi}_{i,j}=2\delta_{j,2N+1-i}\,\alpha_i\frac{p_i^{-2}-1}{p_i^{-1}-p_j^{-1}}
+\frac{(p_i^{-1}-p_j^{-1})(p_i^{-1}+1)}{(p_i^{-1}+p_j^{-1})(p_j^{-1}-1)}e^{\xi_i+\xi_j}$. 
Then we can further simplify as follows:
\begin{equation*}
\tau_2 =
\frac{\prod_{k=1}^{2N}\frac{p_k^{-1}+1}{p_k^{-1}-1}}
{c}\,
\left|
\begin{array}{cccccc}
\tilde{\Psi}_{1,1} & \tilde{\Psi}_{1,2} & \cdots & \tilde{\Psi}_{1,2N} &
\frac{p_1^{-1}-1}{p_1^{-1}+1}e^{\xi_1} & e^{\xi_1} \\
\vdots & \vdots  & \ddots &\vdots &\vdots &\vdots  \\
\tilde{\Psi}_{2N,1} & \tilde{\Psi}_{2N,2}  & \cdots  & \tilde{\Psi}_{2N,2N} & 
\frac{p_{2N}^{-1}-1}{p_{2N}^{-1}+1}e^{\xi_{2N}} 
& e^{\xi_{2N}}  \\
p_1^{-1}e^{\xi_1} & p_2^{-1}e^{\xi_2} & \cdots 
& p_{2N}^{-1}e^{\xi_{2N}} & 1 &0\\
e^{\xi_1} & e^{\xi_2} & \cdots & 
e^{\xi_{2N}} & 1 &1
\end{array}
\right|
\,,
\end{equation*}
where
 $\tilde{\Psi}_{i,j}=2\delta_{j,2N+1-i}\,\alpha_i\frac{(p_i^{-1}-1)(p_j^{-1}-1)}
{p_i^{-1}-p_j^{-1}}
+\frac{p_i^{-1}-p_j^{-1}}{p_i^{-1}+p_j^{-1}}e^{\xi_i+\xi_j}$. 
Note that the $2N \times 2N$ matrix $(\tilde{\Psi}_{i,j})_{1\leq i,j\leq
 2N}$ is
 skew-symmetric and $\tilde{\Psi}_{i,j}=-\Psi_{i,j}$.  
Thus we can use the formula (\ref{formula4}): 
\begin{eqnarray*}
 &&\tau_2
=\frac{\prod_{k=1}^{2N}\frac{p_k^{-1}+1}{p_k^{-1}-1}}
{c}\,\left(
\left.
\begin{array}{cccccccc}
| & \tilde{\Psi}_{1,2} & \tilde{\Psi}_{1,3} &\cdots & 
 \tilde{\Psi}_{1,2N} & e^{\xi_1} & e^{\xi_1}\\
&  & \tilde{\Psi}_{2,3} & \cdots  & \tilde{\Psi}_{2,2N} & e^{\xi_2} & e^{\xi_2}\\
&  &  & \ddots  &\vdots & \vdots &\vdots \\
 &  & &  &   \tilde{\Psi}_{2N-1, 2N} & e^{\xi_{2N-1}} & e^{\xi_{2N-1}} \\
 &  & &  &        & e^{\xi_{2N}} & e^{\xi_{2N}} \\
&  &  & & &    & 1
\end{array}
\right|\right.\\
 && \times \left.
\begin{array}{cccccccc}
| & \tilde{\Psi}_{1,2} & \tilde{\Psi}_{1,3} &\cdots & 
 \tilde{\Psi}_{1,2N} & \frac{p_1^{-1}-1}{p_1^{-1}+1}e^{\xi_1} & p_1^{-1}e^{\xi_1}\\
&  & \tilde{\Psi}_{2,3} & \cdots  & \tilde{\Psi}_{2,2N} &
 \frac{p_2^{-1}-1}{p_2^{-1}+1}e^{\xi_2} 
& p_2^{-1}e^{\xi_2}\\
&  &  & \ddots  &\vdots & \vdots &\vdots \\
 &  & &  &   \tilde{\Psi}_{2N-1, 2N} &
  \frac{p_{2N-1}^{-1}-1}{p_{2N-1}^{-1}+1}e^{\xi_{2N-1}} 
& p_{2N-1}^{-1}e^{\xi_{2N-1}} \\
 &  & &  &        
& \frac{p_{2N}^{-1}-1}{p_{2N}^{-1}+1}e^{\xi_{2N}} & p_{2N}^{-1}e^{\xi_{2N}} \\
&  &  & & &    & 1
\end{array}
\right|\\
 && -
\left.
\begin{array}{cccccccc}
| & \tilde{\Psi}_{1,2} & \tilde{\Psi}_{1,3} &\cdots & 
 \tilde{\Psi}_{1,2N} & e^{\xi_1} & p_1^{-1}e^{\xi_1}\\
&  & \tilde{\Psi}_{2,3} & \cdots  & \tilde{\Psi}_{2,2N} & e^{\xi_2} & p_2^{-1}e^{\xi_2}\\
&  &  & \ddots  &\vdots & \vdots &\vdots \\
 &  & &  &   \tilde{\Psi}_{2N-1, 2N} & e^{\xi_{2N-1}} & p_{2N-1}^{-1}e^{\xi_{2N-1}} \\
 &  & &  &        & e^{\xi_{2N}} & p_{2N}^{-1}e^{\xi_{2N}} \\
&  &  & & &    & 0
\end{array}
\right|\\
 && \left. \times \left.
\begin{array}{cccccccc}
| & \tilde{\Psi}_{1,2} & \tilde{\Psi}_{1,3} &\cdots & 
 \tilde{\Psi}_{1,2N} & \frac{p_1^{-1}-1}{p_1^{-1}+1}e^{\xi_1} & e^{\xi_1}\\
&  & \tilde{\Psi}_{2,3} & \cdots  & \tilde{\Psi}_{2,2N} & 
\frac{p_2^{-1}-1}{p_2^{-1}+1}e^{\xi_2} & e^{\xi_2}\\
&  &  & \ddots  &\vdots & \vdots &\vdots \\
 &  & &  &   \tilde{\Psi}_{2N-1, 2N} 
& \frac{p_{2N-1}^{-1}-1}{p_{2N-1}^{-1}+1}e^{\xi_{2N-1}} & e^{\xi_{2N-1}} \\
 &  & &  &        & \frac{p_{2N}^{-1}-1}{p_{2N}^{-1}+1}e^{\xi_{2N}} & e^{\xi_{2N}} \\
&  &  & & &    & 1
\end{array}
\right|\right)\,.
%%%\\ &=\frac{\prod_{k=1}^{2N}\frac{p_k^{-1}+1}{p_k^{-1}-1}}{c}\,\left[
%%%{\rm pf}(\tilde{\Psi}_{i,j})_{1\leq i,j\leq 2N}\,
%%%{\rm pf}\left(\tilde{\Psi}_{ij}-\frac{p_i^{-1}-1}{p_i^{-1}+1}p_j^{-1}e^{\xi_i+\xi_j}+
%%%p_i^{-1}\frac{p_j^{-1}-1}{p_j^{-1}+1}e^{\xi_i+\xi_j}\right)_{1\leq i,j\leq 2N}\right.\\
%%% & -\left.
%%%\begin{array}{cccccccc}
%%%| & \tilde{\Psi}_{1,2} & \tilde{\Psi}_{1,3} &\cdots & 
%%% \tilde{\Psi}_{1,2N} & e^{\xi_1} & p_1^{-1}e^{\xi_1}\\
%%%&  & \tilde{\Psi}_{2,3} & \cdots  & \tilde{\Psi}_{2,2N} & e^{\xi_2} & p_2^{-1}e^{\xi_2}\\
%%%&  &  & \ddots  &\vdots & \vdots &\vdots \\
%%% &  & &  &   \tilde{\Psi}_{2N-1, 2N} & e^{\xi_{2N-1}} & p_{2N-1}^{-1}e^{\xi_{2N-1}} \\
%%% &  & &  &        & e^{\xi_{2N}} & p_{2N}^{-1}e^{\xi_{2N}} \\
%%%&  &  & & &    & 0
%%%\end{array}
%%%\right|
%%%\\& \left.\times {\rm
%%%	pf}\left(\tilde{\Psi}_{ij}-\frac{p_i^{-1}-1}{p_i^{-1}+1}
%%%e^{\xi_i+\xi_j}+
%%%\frac{p_j^{-1}-1}{p_j^{-1}+1}e^{\xi_i+\xi_j}\right)_{1\leq i,j\leq 2N}\right]\,.
\end{eqnarray*}
Let
\begin{eqnarray*}
\fl  g_1&=& {\rm pf}\left(\Psi_{i,j}\right)_{1\leq i,j\leq
2N}= {\rm pf}\left(\tilde{\Psi}_{i,j}\right)_{1\leq i,j\leq
2N}\\
\fl & =&\left.
\begin{array}{cccccccc}
| & \tilde{\Psi}_{1,2} & \tilde{\Psi}_{1,3} &\cdots & 
 \tilde{\Psi}_{1,2N} & e^{\xi_1} & e^{\xi_1}\\
&  & \tilde{\Psi}_{2,3} & \cdots  & \tilde{\Psi}_{2,2N} & e^{\xi_2} & e^{\xi_2}\\
&  &  & \ddots  &\vdots & \vdots &\vdots \\
 &  & &  &   \tilde{\Psi}_{2N-1, 2N} & e^{\xi_{2N-1}} & e^{\xi_{2N-1}} \\
 &  & &  &        & e^{\xi_{2N}} & e^{\xi_{2N}} \\
&  &  & & &    & 1
\end{array}
\right|\,,
\end{eqnarray*}
\begin{eqnarray*}
\fl  g_2&=&{\rm pf}\left(2\alpha_i\frac{(p_i+1)(p_j+1)}{p_i-p_j}\delta_{j,2N+1-i}
+\frac{p_i-p_j}{p_i+p_j}e^{\xi_i+\xi_j}\right)_{1\leq
i,j\leq 2N}\\
\fl  &=&{\rm pf}\left(2\alpha_i\frac{(p_i^{-1}+1)(p_j^{-1}+1)}{p_i^{-1}-p_j^{-1}}\delta_{j,2N+1-i}
+\frac{p_i^{-1}-p_j^{-1}}{p_i^{-1}+p_j^{-1}}e^{\xi_i+\xi_j}\right)_{1\leq
i,j\leq 2N}\nonumber\\
\fl  &=&\prod_{k=1}^{2N}\frac{p_k^{-1}+1}{p_k^{-1}-1}\,
{\rm pf}\left(\tilde{\Psi}_{ij}-\frac{p_i^{-1}-1}{p_i^{-1}+1}e^{\xi_i+\xi_j}+
\frac{p_j^{-1}-1}{p_j^{-1}+1}e^{\xi_i+\xi_j}\right)_{1\leq i,j\leq 2N}\\
\fl  &=&\prod_{k=1}^{2N}\frac{p_k^{-1}+1}{p_k^{-1}-1}\times
\left.
\begin{array}{cccccccc}
| & \tilde{\Psi}_{1,2} & \tilde{\Psi}_{1,3} &\cdots & 
 \tilde{\Psi}_{1,2N} & \frac{p_1^{-1}-1}{p_1^{-1}+1}e^{\xi_1} & e^{\xi_1}\\
&  & \tilde{\Psi}_{2,3} & \cdots  & \tilde{\Psi}_{2,2N} &
 \frac{p_2^{-1}-1}{p_2^{-1}+1}
e^{\xi_2} & e^{\xi_2}\\
&  &  & \ddots  &\vdots & \vdots &\vdots \\
 &  & &  &   \tilde{\Psi}_{2N-1, 2N} & 
\frac{p_{2N-1}^{-1}-1}{p_{2N-1}^{-1}+1}e^{\xi_{2N-1}} & e^{\xi_{2N-1}} \\
 &  & &  &        & \frac{p_{2N}^{-1}-1}{p_{2N}^{-1}+1}e^{\xi_{2N}} & e^{\xi_{2N}} \\
&  &  & & &    & 1
\end{array}
\right|\,.
\end{eqnarray*}
Then
\begin{equation*}
\fl \partial_{x_{-1}}g_1=\left.
\begin{array}{cccccccc}
| & \tilde{\Psi}_{1,2} & \tilde{\Psi}_{1,3} &\cdots & 
 \tilde{\Psi}_{1,2N} & e^{\xi_1} & p_1^{-1}e^{\xi_1}\\
&  & \tilde{\Psi}_{2,3} & \cdots  & \tilde{\Psi}_{2,2N} & e^{\xi_2} & p_2^{-1}e^{\xi_2}\\
&  &  & \ddots  &\vdots & \vdots &\vdots \\
 &  & &  &   \tilde{\Psi}_{2N-1, 2N} & e^{\xi_{2N-1}} & p_{2N-1}^{-1}e^{\xi_{2N-1}} \\
 &  & &  &        & e^{\xi_{2N}} & p_{2N}^{-1}e^{\xi_{2N}} \\
&  &  & & &    & 0
\end{array}
\right|\,,
\end{equation*}
\begin{eqnarray*}
\fl &&(\partial_{x_{-1}}+1)g_2\\
\fl &&=
(\partial_{x_{-1}}+1){\rm pf}\left(2\alpha_i
\frac{(p_i^{-1}+1)(p_j^{-1}+1)}{p_i^{-1}-p_j^{-1}}\delta_{j,2N+1-i}
+\frac{p_i^{-1}-p_j^{-1}}{p_i^{-1}+p_j^{-1}}e^{\xi_i+\xi_j}\right)_{1\leq
i,j\leq 2N}
\\
\fl  && = 
{\rm pf}\left(2\alpha_i\frac{(p_i^{-1}+1)(p_j^{-1}+1)}{p_i^{-1}-p_j^{-1}}\delta_{j,2N+1-i}
+\frac{p_i^{-1}-p_j^{-1}}{p_i^{-1}+p_j^{-1}}e^{\xi_i+\xi_j}+(p_i^{-1}-p_j^{-1}
)e^{\xi_i+\xi_j}\right)_{1\leq
i,j\leq 2N}
\\
\fl  && = \prod_{k=1}^{2N}\frac{p_k^{-1}+1}{p_k^{-1}-1}\,
{\rm pf}\left(\tilde{\Psi}_{ij}-\frac{p_i^{-1}-1}{p_i^{-1}+1}p_j^{-1}e^{\xi_i+\xi_j}+
p_i^{-1}\frac{p_j^{-1}-1}{p_j^{-1}+1}e^{\xi_i+\xi_j}\right)_{1\leq i,j\leq 2N}\nonumber
\\
\fl  &&=\prod_{k=1}^{2N}\frac{p_k^{-1}+1}{p_k^{-1}-1}\times 
\left.
\begin{array}{cccccccc}
| & \tilde{\Psi}_{1,2} & \tilde{\Psi}_{1,3} &\cdots & 
 \tilde{\Psi}_{1,2N} & \frac{p_1^{-1}-1}{p_1^{-1}+1}e^{\xi_1} & p_1^{-1}e^{\xi_1}\\
&  & \tilde{\Psi}_{2,3} & \cdots  & \tilde{\Psi}_{2,2N} & 
\frac{p_2^{-1}-1}{p_2^{-1}+1}e^{\xi_2} & p_2^{-1}e^{\xi_2}\\
&  &  & \ddots  &\vdots & \vdots &\vdots \\
 &  & &  &   \tilde{\Psi}_{2N-1, 2N} 
& \frac{p_{2N-1}^{-1}-1}{p_{2N-1}^{-1}+1}e^{\xi_{2N-1}} & p_{2N-1}^{-1}e^{\xi_{2N-1}} \\
 &  & &  &        & \frac{p_{2N}^{-1}-1}{p_{2N}^{-1}+1}e^{\xi_{2N}} & p_{2N}^{-1}e^{\xi_{2N}} \\
&  &  & & &    & 1
\end{array}
\right|\,.
\end{eqnarray*}
Thus $\tau_2$, $g_1$ and $g_2$ satisfy the relation (\ref{tau2-g1-g2}). 
\end{proof}

Letting $F=\tau_{0}$, $G=\tau_{1}$ and $H=\tau_{2}$, 
we obtain the following equations 
\begin{eqnarray}
&&-\left(\frac{1}{2}D_{x_1}D_{x_{-1}}-1\right)F\cdot F=G^2\,,\label{dp-bilinear-1}\\
&&-\left(\frac{1}{2}D_{x_1}D_{x_{-1}}-1\right)G\cdot G=FH\,,\label{dp-bilinear-2}\\
&&cG=g_1g_2\,,\label{dp-relation-1}\\
&&cF=g_1g_2-D_{x_{1}}g_1\cdot g_2\,,\label{dp-relation-2}\\
&&cH=g_1g_2-D_{x_{-1}}g_1\cdot g_2\label{dp-relation-3}\,,
\end{eqnarray}
from the bilinear equations 
(\ref{3CToda-bilinear1}) and 
(\ref{3CToda-bilinear2}), and the relations between determinants and
pfaffians 
(\ref{tau1-g1-g2}), (\ref{tau0-g1-g2}), (\ref{tau2-g1-g2}). 
   
\begin{theorem}\label{theorem:hodograph}
The $\tau$-functions $g_1$ (\ref{g1}) and $g_2$ (\ref{g2}) of 
equations (\ref{dp-bilinear-1}), (\ref{dp-bilinear-2}),
 (\ref{dp-relation-1}), (\ref{dp-relation-2}) and (\ref{dp-relation-3})  
give the $N$-soliton solution of the DP equation
\begin{equation}
u_t+3u_x- u_{txx}+4uu_x=3u_xu_{xx}+uu_{xxx}\,,\label{DP-eq-nokappa}
\end{equation} 
through the dependent variable transformation
\[
u=-\left(\ln \frac{g_1}{g_2}\right)_{x_{-1}}\,, 
\] 
and the hodograph (reciprocal) transformation
\begin{equation}
\left\{
\begin{array}{c}
x=x_{1}+\int_{-\infty}^{x_{-1}}u(x_1,x_{-1}')dx_{-1}'\\
=x_{1}-\ln \frac{g_1}{g_2}\,,\qquad \qquad \quad  \\
t=x_{-1} \,.\qquad \qquad \qquad \qquad  \quad 
\end{array}
\right.\label{hodograph}
\end{equation}
\end{theorem}

\begin{proof}
From (\ref{dp-relation-1}), (\ref{dp-relation-2}) and
 (\ref{dp-relation-3}), 
we have the relations
\begin{eqnarray}
&&-\left(\ln
    \frac{g_1}{g_2}\right)_{x_{1}}=\frac{F}{G}-1\,,\label{relation-g-G-F}\\
&&-\left(\ln
    \frac{g_1}{g_2}\right)_{x_{-1}}=\frac{H}{G}-1\,.\label{relation-g-G-H}
\end{eqnarray}
Let 
\begin{equation}
\rho=\frac{G}{F}\,,\quad 
u=-\left(\ln \frac{g_1}{g_2}\right)_{x_{-1}}\,.
\end{equation} 
Differentiating (\ref{relation-g-G-F}) with respect to $x_{-1}$, we obtain
\begin{equation}
u_{x_{1}}=\left(\frac{1}{\rho}\right)_{x_{-1}}\,.\label{u-rho-relation}
\end{equation}
This is rewritten as 
\begin{equation}
(\ln \rho)_{x_{-1}}=-\rho u_{x_{1}}\,.\label{rho-u}
\end{equation}
Equation (\ref{relation-g-G-H}) leads to 
\begin{equation}
\frac{H}{G}=1+u\,.
\end{equation}

The bilinear equations (\ref{dp-bilinear-1}) and (\ref{dp-bilinear-2}) 
are written as 
\begin{eqnarray}
&&-(\ln F)_{x_1x_{-1}}+1=\rho^2\,,\label{ln-bi-1}\\
&&-(\ln G)_{x_1x_{-1}}+1=\frac{1}{\rho}(1+u)\,.\label{ln-bi-2}
\end{eqnarray}
Subtracting (\ref{ln-bi-1}) from (\ref{ln-bi-2}), we obtain 
\begin{equation}
-(\ln \rho)_{x_1x_{-1}}=\frac{1}{\rho}(1+u)-\rho^2\,,\label{tzitzeica}
\end{equation}
which leads to 
\begin{equation}
\rho^3=1+u+\rho (\ln \rho)_{x_{1}x_{-1}}\,.
\end{equation}
Using (\ref{rho-u}), it becomes
\begin{equation}
\rho^3=1+u-\rho (\rho u_{x_{1}})_{x_{1}}\,.\label{rho-relation}
\end{equation}

Let us consider the hodograph (reciprocal) transformation
\begin{equation}
\left\{
\begin{array}{l}
x=x_{1}+\int_{-\infty}^{x_{-1}}u(x_1,x_{-1}')dx_{-1}'\\
\quad =x_{1}-\ln \frac{g_1}{g_2}\,,\qquad \qquad \\
t=x_{-1} \,.\qquad \qquad \qquad \qquad  
\end{array}
\right.\label{hodograph1}
\end{equation}
This yields
\begin{equation}
\left\{
\begin{array}{c}
\frac{\partial x}{\partial x_{1}}=1-\left(\ln \frac{g_1}{g_2}\right)_{x_{1}}=\rho^{-1}\,,\\
\frac{\partial x}{\partial x_{-1}}=-\left(\ln \frac{g_1}{g_2}\right)_{x_{-1}}=u\,,\quad \\
\end{array}
\right.
\end{equation}
and
\begin{equation}
\left\{
\begin{array}{c}
\partial_{x_{1}}=\frac{1}{\rho}\partial_x\,,\quad\\
\partial_{x_{-1}}=\partial_t+u\partial_x\,.\\
\end{array}
\right.\label{hodograph-derivative}
\end{equation}
Applying the hodograph (reciprocal) 
transformation to (\ref{rho-u}) and 
(\ref{rho-relation}), 
we obtain 
\begin{equation}
\left\{
\begin{array}{c}
(\partial_t+u\partial_x)\ln\rho=-u_x\,,\\
\rho^3=1+u-u_{xx}\,.\qquad \quad 
\end{array}
\right.\label{vakhnenko-form2}
\end{equation}
This is equivalent to
\begin{equation}
(\partial_t+u\partial_x)\ln(1+u-u_{xx})=-3u_x\,,\label{short-DP-eq3}
\end{equation}
which can be written as 
\begin{equation}
(\partial_t+u\partial_x)(1+u-u_{xx})=-3u_x(1+u-u_{xx})\,.\label{short-DP-eq2}
\end{equation}
This is nothing but the DP equation (\ref{DP-eq-nokappa}). 
\end{proof}

\begin{remark}
Applying the scale transformation $u\to \frac{1}{\kappa^3}u$, $t\to
 \kappa^3t$ to (\ref{DP-eq-nokappa}), we
 obtain the DP equation (\ref{DP-eq}). 
\end{remark}

\begin{remark}
Setting $u=0$ in (\ref{tzitzeica}), we obtain the Tzitzeica equation~
\cite{Tzitzeica1,Tzitzeica2,Nimmo-Ruijsenaars,Willox} 
\begin{equation}
(\ln \rho)_{x_1x_{-1}}=\rho^2-\frac{1}{\rho}\,.
\end{equation}
Thus (\ref{tzitzeica}) and (\ref{u-rho-relation}) 
can be considered as an extension 
of the Tzitzeica equation. 
\end{remark}

Let $k_i=p_i+p_{2N+1-i}$. From 
$p_i^2-p_ip_{2N+1-i}+p_{2N+1-i}^2=1$, we obtain
$p_i=\frac{1}{6}(3k_i+\sqrt{3(4-k_i^2)})$, 
$p_{2N+1-i}=\frac{1}{6}(3k_i-\sqrt{3(4-k_i^2)})$, 
$p_ip_{2N+1-i}=\frac{k_i^2-1}{3}$ and
$\frac{1}{p{i}}+\frac{1}{p_{2N+1-i}}=\frac{3k_i}{k_i^2-1}$. 
Thus 
\[
\xi_i+\xi_{2N+1-i}
=k_ix_1+\frac{3k_i}{k_i^2-1}x_{-1}+\xi_{i0}+\xi_{{2N+1-i}0}\,.
\]
In the pfaffian solution, all phase functions can be expressed by the
summation of $\xi_i+\xi_{2N+1-i}$. So the phase functions can be
expressed by the parameters $\{k_i\}$ ($i=1,2,\cdots,N$). 
Each coefficient of exponential functions can be normalized to 1 after
absorption into phase constants or can be rewritten by the parameters
$\{k_i\}$. 
Thus it is possible to rewrite the above $\tau$-function by using the
parameters $\{k_i\}$ instead of $\{p_i\}$.\vspace{0.2in}\\

\begin{example} 
Soliton Solutions\\ 
For $N=1$, 
\begin{eqnarray*}
g_1&=&2\alpha_1\frac{(p_1-1)(p_2-1)}{p_1-p_2}+\frac{p_1-p_2}{p_1+p_2}e^{\xi_1+\xi_{2}}\\
&=&\frac{(p_1-1)(p_2-1)}{p_1-p_2}\left(2\alpha_1+
\frac{(p_1-p_2)^2}{(p_1+p_2)(p_1-1)(p_2-1)}e^{\xi_1+\xi_{2}}
\right)
\,.
\end{eqnarray*}
\begin{eqnarray*}
g_2&=&2\alpha_1\frac{(p_1+1)(p_2+1)}{p_1-p_2}+\frac{p_1-p_2}{p_1+p_2}e^{\xi_1+\xi_{2}}\\
&=&\frac{(p_1+1)(p_2+1)}{p_1-p_2}\left(2\alpha_1+
\frac{(p_1-p_2)^2}{(p_1+p_2)(p_1+1)(p_2+1)}e^{\xi_1+\xi_{2}}
\right)
\,.
\end{eqnarray*}
Letting $\alpha_1=\frac{1}{2}$ and 
$e^{\gamma_1}=\frac{(p_1-p_2)^2}{(p_1+p_2)}\frac{1}{\sqrt{(p_1-1)(p_2-1)(p_1+1)(p_2+1)}}$, 
the $\tau$-functions can be rewritten as 
\begin{equation*}
g_1=1+e^{\xi_1+\xi_{2}+\phi_1+\gamma_1}\,, \quad 
g_2=1+e^{\xi_1+\xi_{2}-\phi_1+\gamma_1}
\end{equation*}
where 
\[
\xi_1+\xi_{2}
=k_1x_1+\frac{3k_1}{k_1^2-1}x_{-1}+\xi_{10}+\xi_{{2}0}\,,
\]
and 
\[
e^{\phi_1}=\sqrt{\frac{(p_1+1)(p_2+1)}{(p_1-1)(p_2-1)}}=\sqrt{\frac{k_1^2+3k_1+2}{k_1^2-3k_1+2}}
=\sqrt{\frac{(k_1 +2)(k_1 +1)}{
(k_1 -2)(k_1 -1)}}\,.
\]
Here $\gamma_1$ can be absorbed into a phase constant. 
\vspace{0.2in} 

\noindent
For $N=2$, 
\begin{eqnarray*}
\fl &&g_1
 =\frac{p_1-p_2}{p_1+p_2}e^{\xi_1+\xi_{2}}
 \frac{p_3-p_4}{p_3+p_4}e^{\xi_3+\xi_{4}}
-\frac{p_1-p_3}{p_1+p_3}e^{\xi_1+\xi_{3}}
 \frac{p_2-p_4}{p_2+p_4}e^{\xi_2+\xi_{4}}\\
\fl  &&\qquad +\left(2\alpha_1\frac{(p_1-1)(p_4-1)}{p_1-p_4}
+\frac{p_1-p_4}{p_1+p_4}e^{\xi_1+\xi_{4}}\right)
\left(2\alpha_2\frac{(p_2-1)(p_3-1)}{p_2-p_3}+\frac{p_2-p_3}{p_2+p_3}e^{\xi_2+\xi_{3}}\right)\\
\fl  &&\quad =\frac{(p_1-1)(p_4-1)}{p_1-p_4}
\cdot \frac{(p_2-1)(p_3-1)}{p_2-p_3}\left(
2\alpha_1
\cdot 2\alpha_2
+2\alpha_2\frac{(p_1-p_4)^2}{(p_1+p_4)(p_1-1)(p_4-1)}e^{\xi_1+\xi_{4}}\right.\\
\fl 
&&\qquad +2\alpha_1\frac{(p_2-p_3)^2}{(p_2+p_3)(p_2-1)(p_3-1)}e^{\xi_2+\xi_{3}}
+\frac{p_1-p_4}{(p_1-1)(p_4-1)}\cdot \frac{p_2-p_3}{(p_2-1)(p_3-1)} \\
\fl &&\qquad  \times \left. \left(\frac{p_1-p_2}{p_1+p_2}
\frac{p_3-p_4}{p_3+p_4}
-\frac{p_1-p_3}{p_1+p_3}\frac{p_2-p_4}{p_2+p_4}
+\frac{p_1-p_4}{p_1+p_4}\frac{p_2-p_3}{p_2+p_3}
\right)e^{\xi_1+\xi_{2}+\xi_3+\xi_{4}}\right)\,,
\end{eqnarray*}
\begin{eqnarray*}
\fl &&g_2
 =\frac{p_1-p_2}{p_1+p_2}e^{\xi_1+\xi_{2}}
 \frac{p_3-p_4}{p_3+p_4}e^{\xi_3+\xi_{4}}
-\frac{p_1-p_3}{p_1+p_3}e^{\xi_1+\xi_{3}}
 \frac{p_2-p_4}{p_2+p_4}e^{\xi_2+\xi_{4}}\\
\fl &&\qquad +\left(2\alpha_1\frac{(p_1+1)(p_4+1)}{p_1-p_4}
+\frac{p_1-p_4}{p_1+p_4}e^{\xi_1+\xi_{4}}\right)
\left(2\alpha_2\frac{(p_2+1)(p_3+1)}{p_2-p_3}+\frac{p_2-p_3}{p_2+p_3}e^{\xi_2+\xi_{3}}\right)\\
\fl &&\quad =\frac{(p_1+1)(p_4+1)}{p_1-p_4}
\cdot \frac{(p_2+1)(p_3+1)}{p_2-p_3}\left(
2\alpha_1
\cdot 2\alpha_2
+2\alpha_2\frac{(p_1-p_4)^2}{(p_1+p_4)(p_1+1)(p_4+1)}e^{\xi_1+\xi_{4}}\right.\\
\fl &&\qquad +2\alpha_1\frac{(p_2-p_3)^2}{(p_2+p_3)(p_2+1)(p_3+1)}e^{\xi_2+\xi_{3}}
+\frac{p_1-p_4}{(p_1+1)(p_4+1)}\cdot \frac{p_2-p_3}{(p_2+1)(p_3+1)} \\
\fl &&\qquad  \times \left. \left(\frac{p_1-p_2}{p_1+p_2}
\frac{p_3-p_4}{p_3+p_4}
-\frac{p_1-p_3}{p_1+p_3}\frac{p_2-p_4}{p_2+p_4}
+\frac{p_1-p_4}{p_1+p_4}\frac{p_2-p_3}{p_2+p_3}
\right)e^{\xi_1+\xi_{2}+\xi_3+\xi_{4}}\right)\,.
\end{eqnarray*}
Letting $\alpha_1=\alpha_2=\frac{1}{2}$,
$e^{\gamma_1}=\frac{(p_1-p_4)^2}{(p_1+p_4)}\frac{1}{\sqrt{(p_1-1)(p_4-1)(p_1+1)(p_4+1)}}$,\\
$e^{\gamma_2}=\frac{(p_2-p_3)^2}{(p_2+p_3)}\frac{1}{\sqrt{(p_2-1)(p_3-1)(p_2+1)(p_3+1)}}$, 
the above $\tau$-functions 
become
\begin{eqnarray*}
g_1&=&1+e^{\xi_1+\xi_{4}+\phi_1}
+e^{\xi_2+\xi_{3}+\phi_2}
+b_{12}e^{\xi_1+\xi_{2}+\xi_3+\xi_{4}+\phi_1+\phi_2}\,,\\
g_2&=&1+e^{\xi_1+\xi_{4}-\phi_1}
+e^{\xi_2+\xi_{3}-\phi_2}
+b_{12}e^{\xi_1+\xi_{2}+\xi_3+\xi_{4}-\phi_1-\phi_2}\,,
\end{eqnarray*}
where 
\begin{eqnarray*}
\fl &&b_{12}= 
\frac{p_1-p_2}{p_1+p_2}
\frac{p_3-p_4}{p_3+p_4}
\frac{p_1+p_4}{p_1-p_4}\frac{p_2+p_3}{p_2-p_3}
-\frac{p_1-p_3}{p_1+p_3}\frac{p_2-p_4}{p_2+p_4}
\frac{p_1+p_4}{p_1-p_4}\frac{p_2+p_3}{p_2-p_3}
+1\\
\fl &&\quad\,  
=\frac{(k_1-k_2)^2(k_1^2-k_1k_2+k_2^2-3)}
{(k_1+k_2)^2(k_1^2+k_1k_2+k_2^2-3)}\,, \\
\fl &&\xi_1+\xi_{4}
=k_1x_1+\frac{3k_1}{k_1^2-1}x_{-1}+\xi_{10}+\xi_{{4}0}\,,\quad 
\xi_2+\xi_{3}
=k_2x_1+\frac{3k_2}{k_2^2-1}x_{-1}+\xi_{20}+\xi_{{3}0}\,,\\
\fl &&e^{\phi_1}=\sqrt{\frac{(p_1+1)(p_4+1)}{(p_1-1)(p_4-1)}}
=\sqrt{\frac{k_1^2+3k_1+2}{k_1^2-3k_1+2}}
=\sqrt{\frac{(k_1 +2)(k_1 +1)}{
(k_1 -2)(k_1 -1)}}
\,,\\
\fl && e^{\phi_2}=\sqrt{\frac{(p_2+1)(p_3+1)}{(p_2-1)(p_3-1)}}
=\sqrt{\frac{k_2^2+3k_2+2}{k_2^2-3k_2+2}}
=\sqrt{\frac{(k_2 +2)(k_2 +1)}{
(k_2 -2)(k_2 -1)}}
\,. 
\end{eqnarray*}
Here $\gamma_1$ and $\gamma_2$ were absorbed into phase constants. 

The $N$-soliton solution of (\ref{DP-eq-nokappa}) is written in the following form: 
\begin{eqnarray*}
&&g_1=\sum_{\mu=0,1}\exp\left[\sum_{i=1}^N\mu_i(\eta_i+\phi_i)
+\sum_{i<j}^{(N)}\mu_i\mu_j\ln b_{ij} \right]\,,\\
&&g_2=\sum_{\mu=0,1}\exp\left[\sum_{i=1}^N\mu_i(\eta_i-\phi_i)
+\sum_{i<j}^{(N)}\mu_i\mu_j\ln b_{ij} \right]\,,\nonumber\\
&&b_{ij}=
\frac{(k_i-k_j)^2(k_i^2-k_ik_j+k_j^2-3)}{(k_i+k_j)^2(k_i^2+k_ik_j+k_j^2-3)}\,
\quad {\rm for}\quad 
i<j\,,\nonumber\\
&&\eta_i=\xi_i+\xi_{2N+1-i}
=k_ix_1+\frac{3k_i}{k_i^2-1}x_{-1}+\eta_{i0}\,,\\
&&e^{\phi_i}=\sqrt{\frac{k_i^2+3k_i+2}{k_i^2-3k_i+2}}
=\sqrt{\frac{(k_i +2)(k_i +1)}{
(k_i -2)(k_i -1)}}
\,, \nonumber
\end{eqnarray*}
where $\displaystyle \sum_{\mu=0,1}$ means the summation over all possible
combinations of $\mu_i=0$ or $1$ for $i=1,2,\cdots,N$, 
and $\displaystyle \sum_{i<j}^{(N)}$ means the summation over all possible
combinations 
of $N$ elements under the condition $i<j$. 
\end{example}

Applying $u\to \frac{1}{\kappa^3}u$, $t\to \kappa^3t$, $x_{-1}\to
 \kappa^3x_{-1}$, $x_1\to \frac{x_1}{\kappa}$, $\kappa k_i=p_i+p_{2N+1-i}$, we
 obtain the $N$-soliton solution of the DP equation (\ref{DP-eq}).
\begin{theorem}\label{theorem:matsuno}
The $N$-soliton solution of the DP equation (\ref{DP-eq}) is given as follows:
\begin{eqnarray*}
&&u=-\left(\ln \frac{g_1}{g_2}\right)_{x_{-1}}\,, \\
&&g_1=\sum_{\mu=0,1}\exp\left[\sum_{i=1}^N\mu_i(\eta_i+\phi_i)
+\sum_{i<j}^{(N)}\mu_i\mu_j\ln b_{ij} \right]\,,\\
&&g_2=\sum_{\mu=0,1}\exp\left[\sum_{i=1}^N\mu_i(\eta_i-\phi_i)
+\sum_{i<j}^{(N)}\mu_i\mu_j\ln b_{ij} \right]\,,\nonumber\\
&&b_{ij}=
\frac{(k_i-k_j)^2((k_i^2-k_ik_j+k_j^2)\kappa^2-3)}
{(k_i+k_j)^2((k_i^2+k_ik_j+k_j^2)\kappa^2-3)}\,
\quad {\rm for}\quad 
i<j\,,\nonumber\\
&&\eta_i=\xi_i+\xi_{2N+1-i}
=k_ix_1+\frac{3k_i\kappa^4}{\kappa^2k_i^2-1}x_{-1}+\eta_{i0}\,,\\
&&e^{\phi_i}=\sqrt{\frac{\kappa^2k_i^2+3\kappa k_i+2}{\kappa^2k_i^2-3\kappa
 k_i+2}}=
\sqrt{\frac{(\kappa k_i +2)(\kappa k_i +1)}{
(\kappa k_i -2)(\kappa k_i -1)}}
\,, 
\nonumber
\end{eqnarray*}
and the hodograph (reciprocal) transformation
\begin{equation*}
\left\{
\begin{array}{c}
x=\frac{x_{1}}\kappa+\int_{-\infty}^{x_{-1}}u(x_1,x_{-1}')dx_{-1}'\\
=\frac{x_{1}}\kappa-\ln \frac{g_1}{g_2}\,,\qquad \qquad \quad  \\
t=x_{-1} \,.\qquad \qquad \qquad \qquad  \quad 
\end{array}
\right.
\end{equation*}
\end{theorem}
This is consistent with the result in \cite{Matsuno-DP1,Matsuno-DP2}. 

\begin{remark}
There are 3 regions in which the above soliton solution becomes regular: 
(i) $\frac{2}{\kappa}<k_i$, (ii) 
$-\frac{1}{\kappa}<k_i< \frac{1}{\kappa}$, 
(iii) $k_i<-\frac{2}{\kappa}$. 
(Note that this is obtained by the reality condition of $e^{\phi_i}$.)  
In the region (ii), the graph of the 
soliton solution shows smooth solitons. 
In the region (i) and (iii), the graph of the soliton solution 
shows loop solitons.  
\end{remark}

In Figure 1-4, we show examples of 2-soliton interactions. 

\begin{figure}[!htbp]
\begin{center}
\includegraphics[width=5cm]{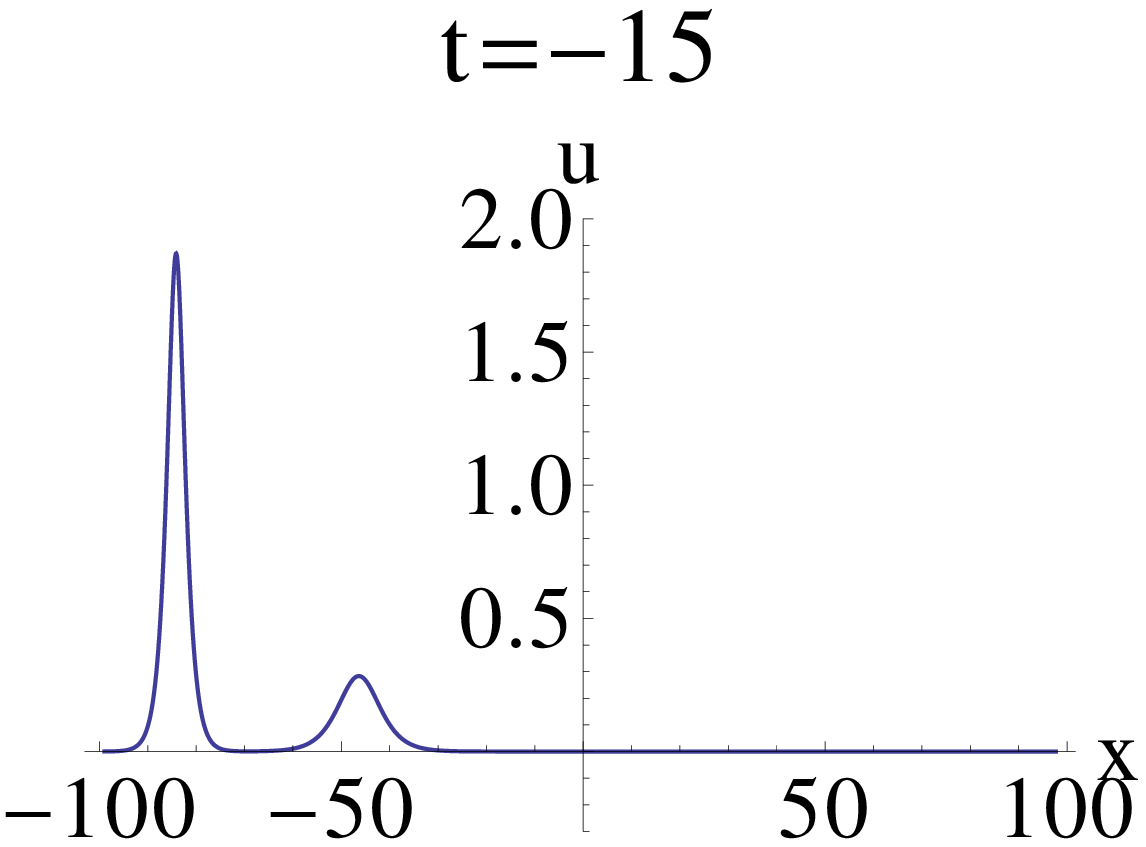}
\includegraphics[width=5cm]{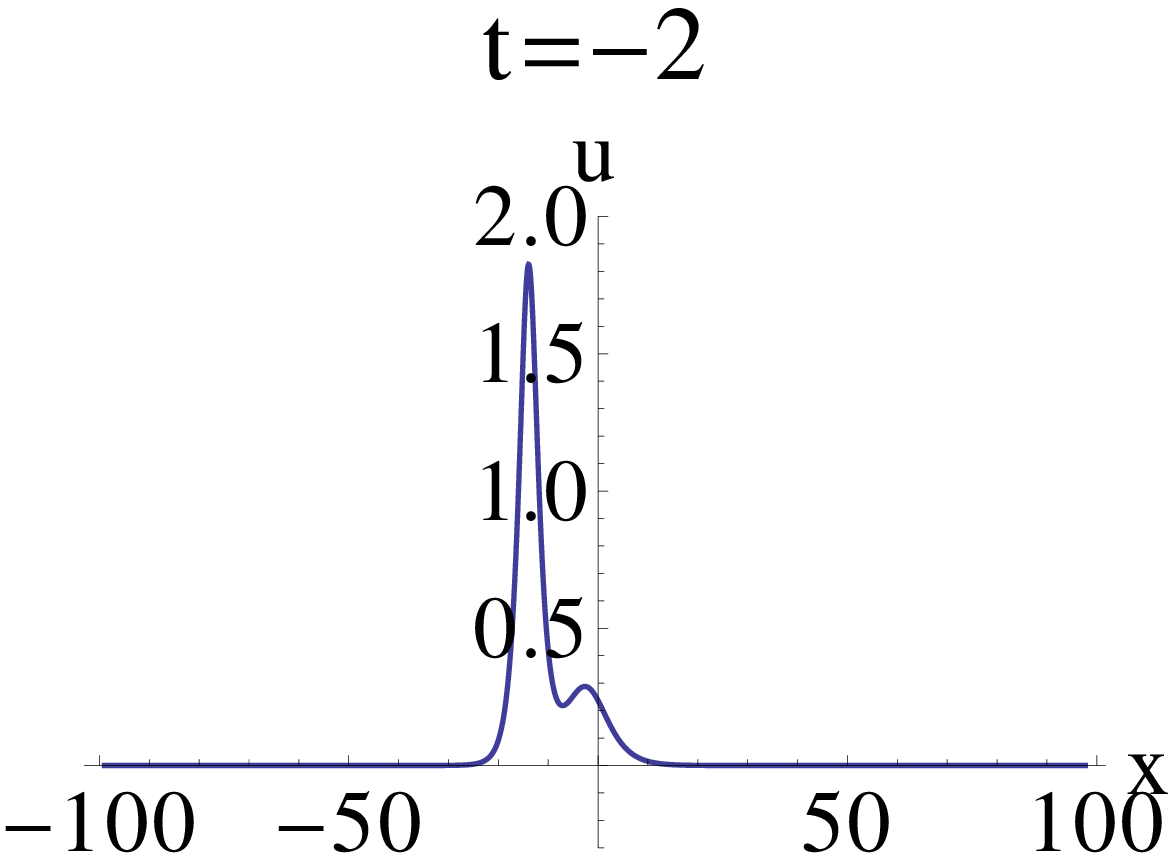}
\includegraphics[width=5cm]{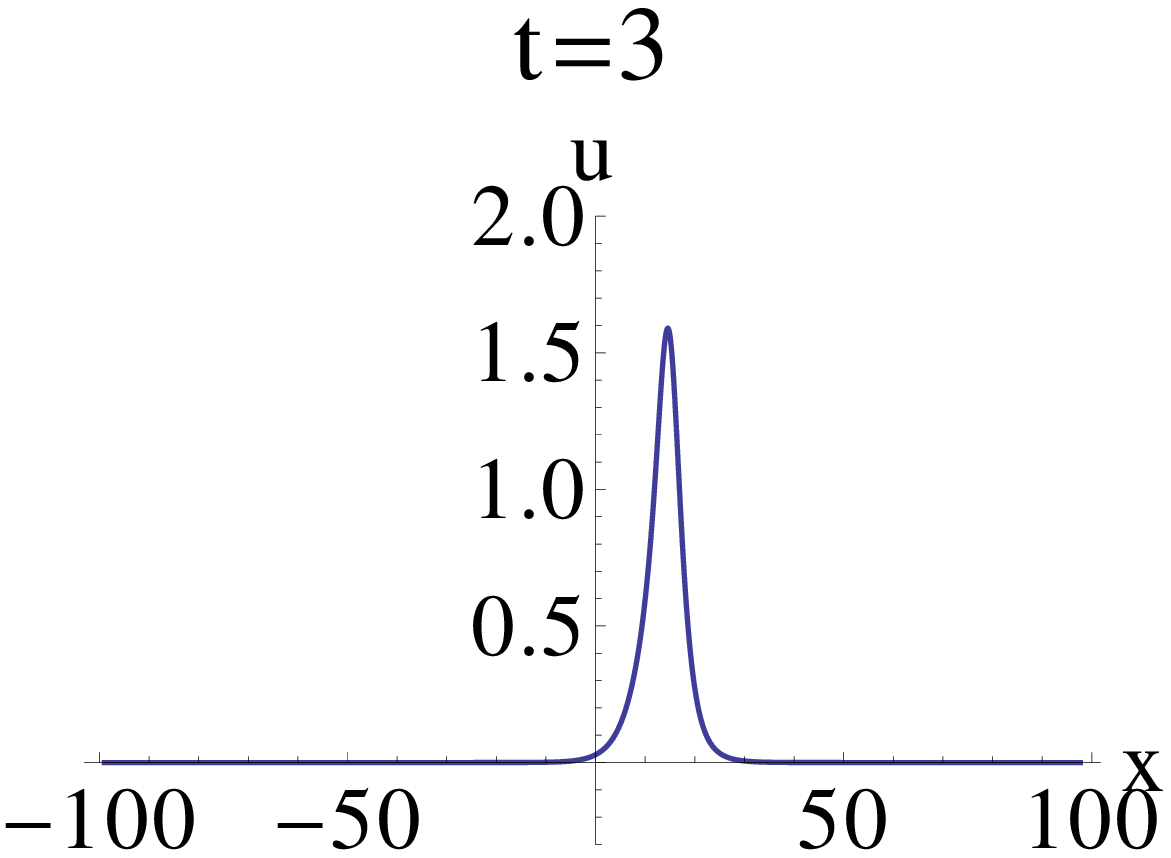}\\
\includegraphics[width=5cm]{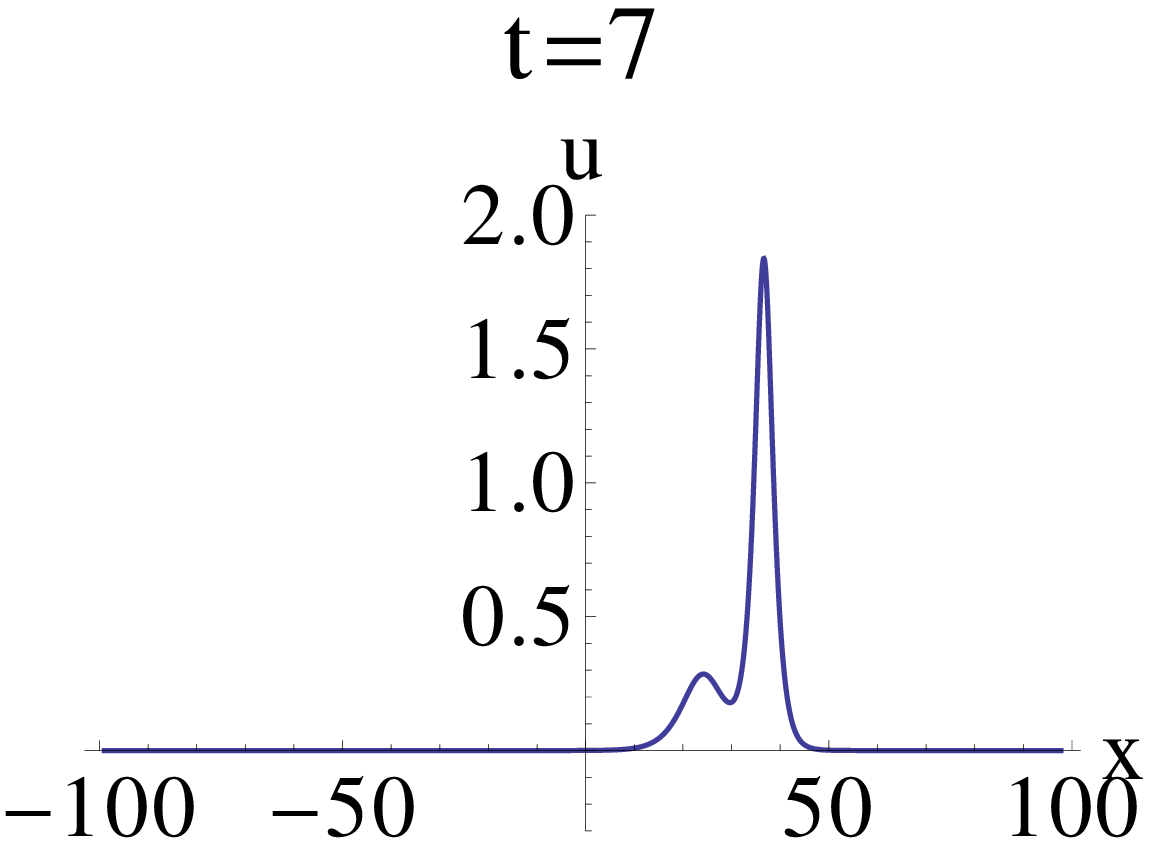}
\includegraphics[width=5cm]{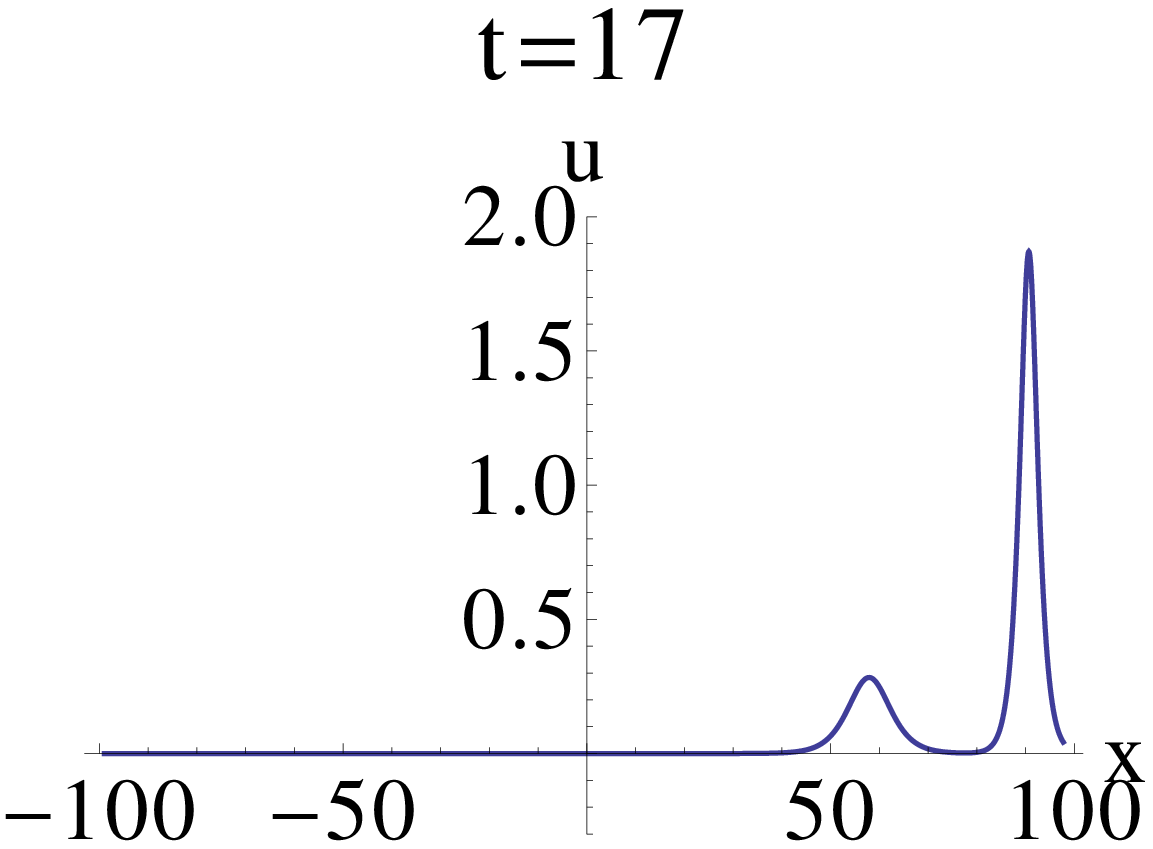}
\end{center}
\caption{2-soliton interaction. 
$\kappa=1$, $k_1=-\frac{1}{3}$, $k_2=\frac{2}{3}$.}
\end{figure}

\begin{figure}[!htbp]
\begin{center}
\includegraphics[width=5cm]{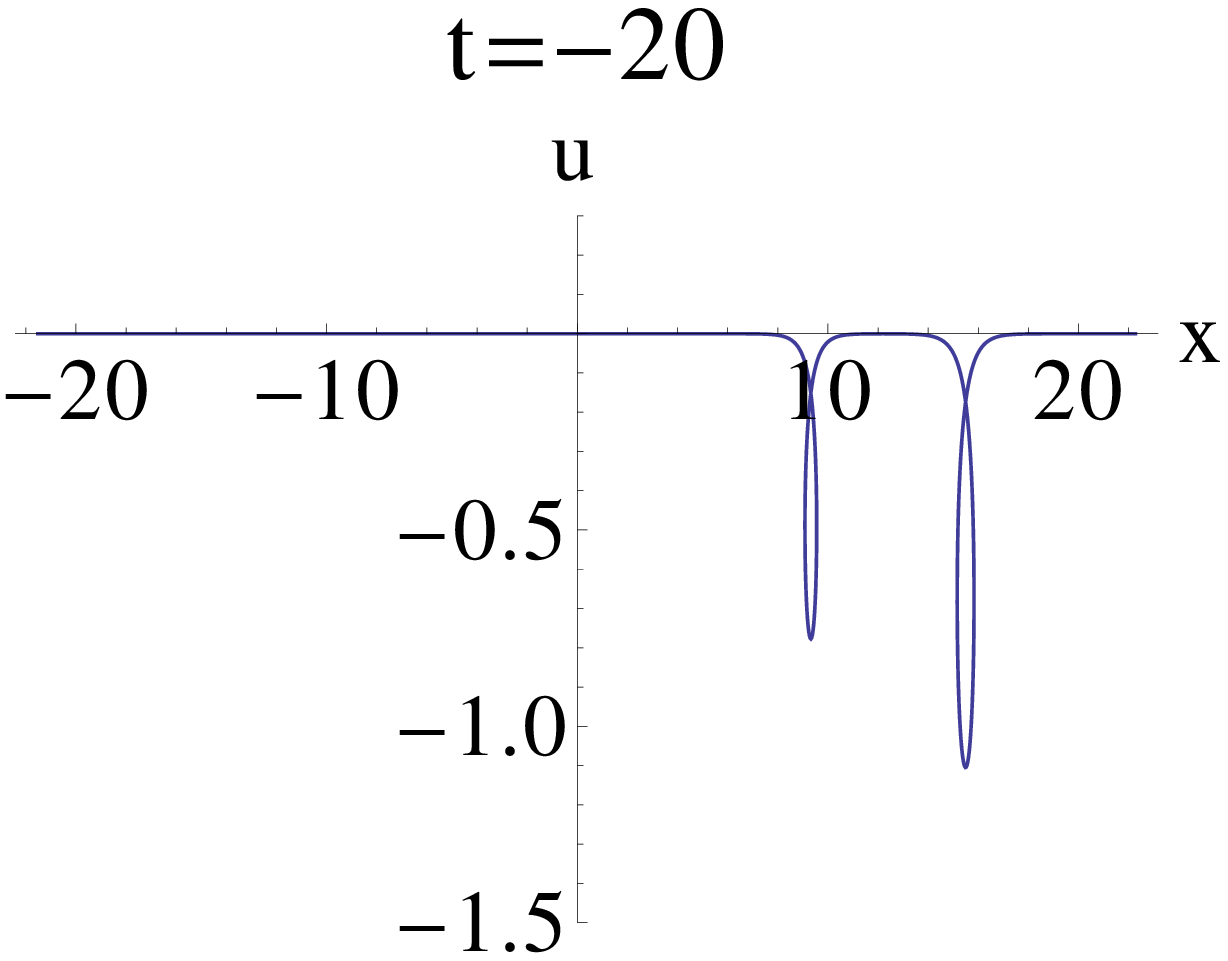}
\includegraphics[width=5cm]{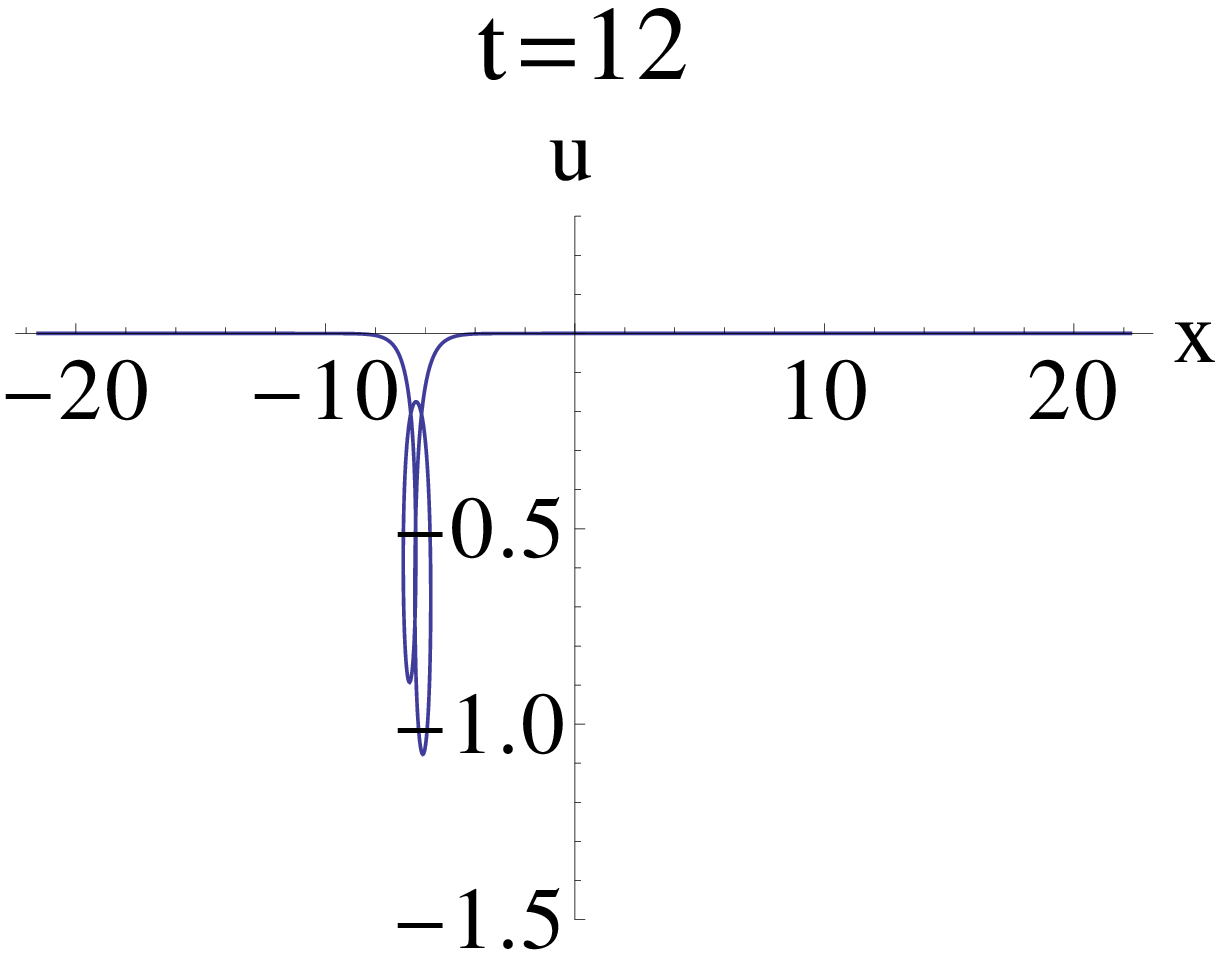}
\includegraphics[width=5cm]{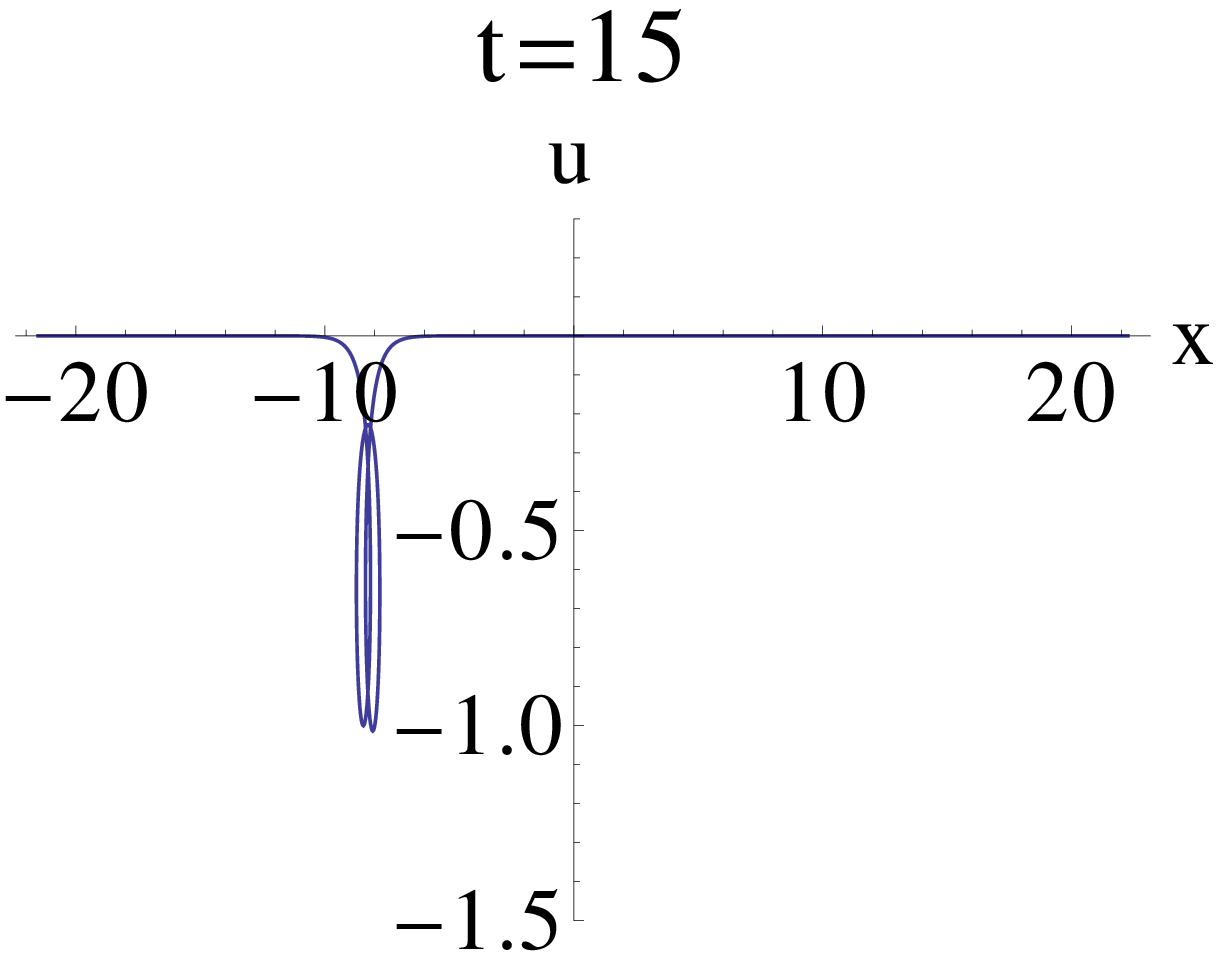}\\
\includegraphics[width=5cm]{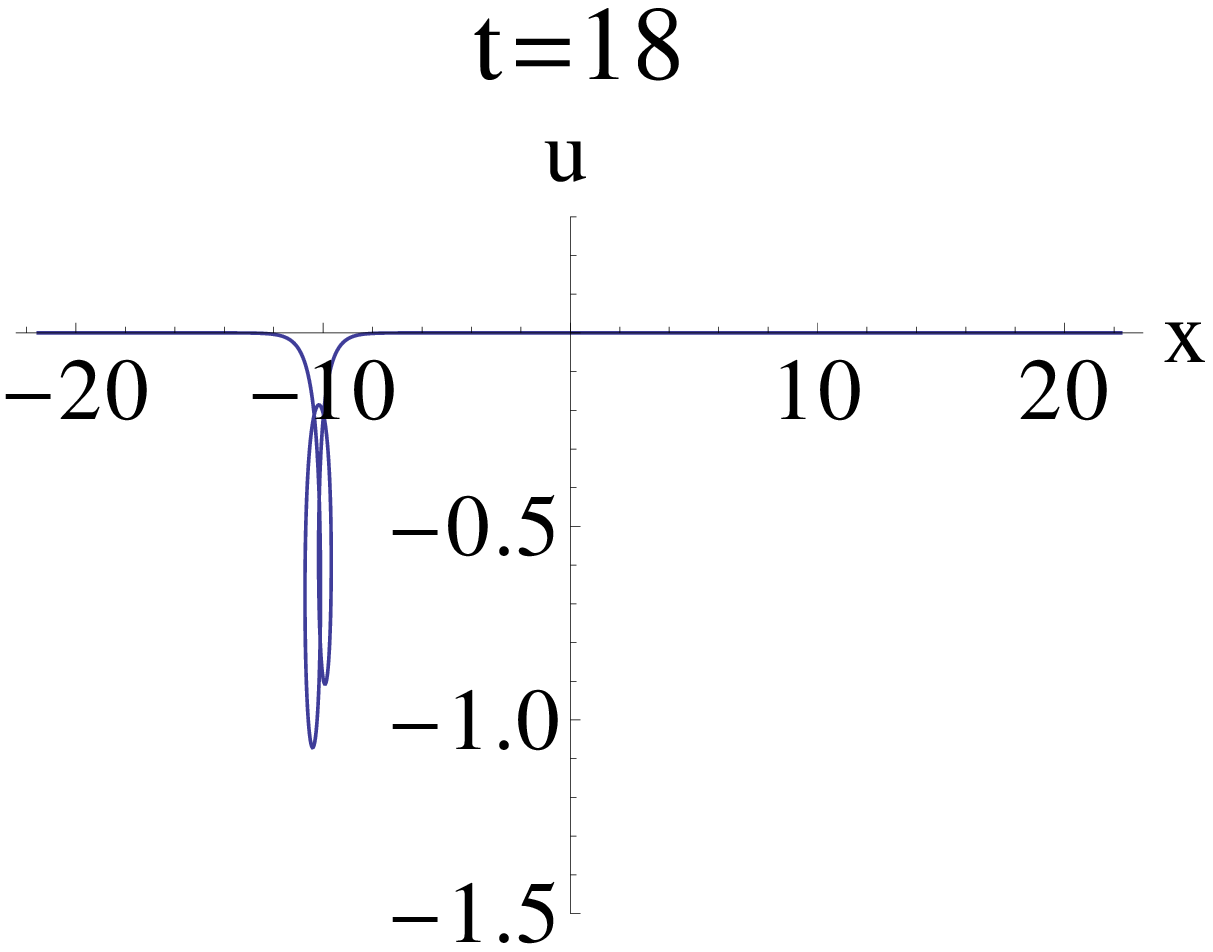}
\includegraphics[width=5cm]{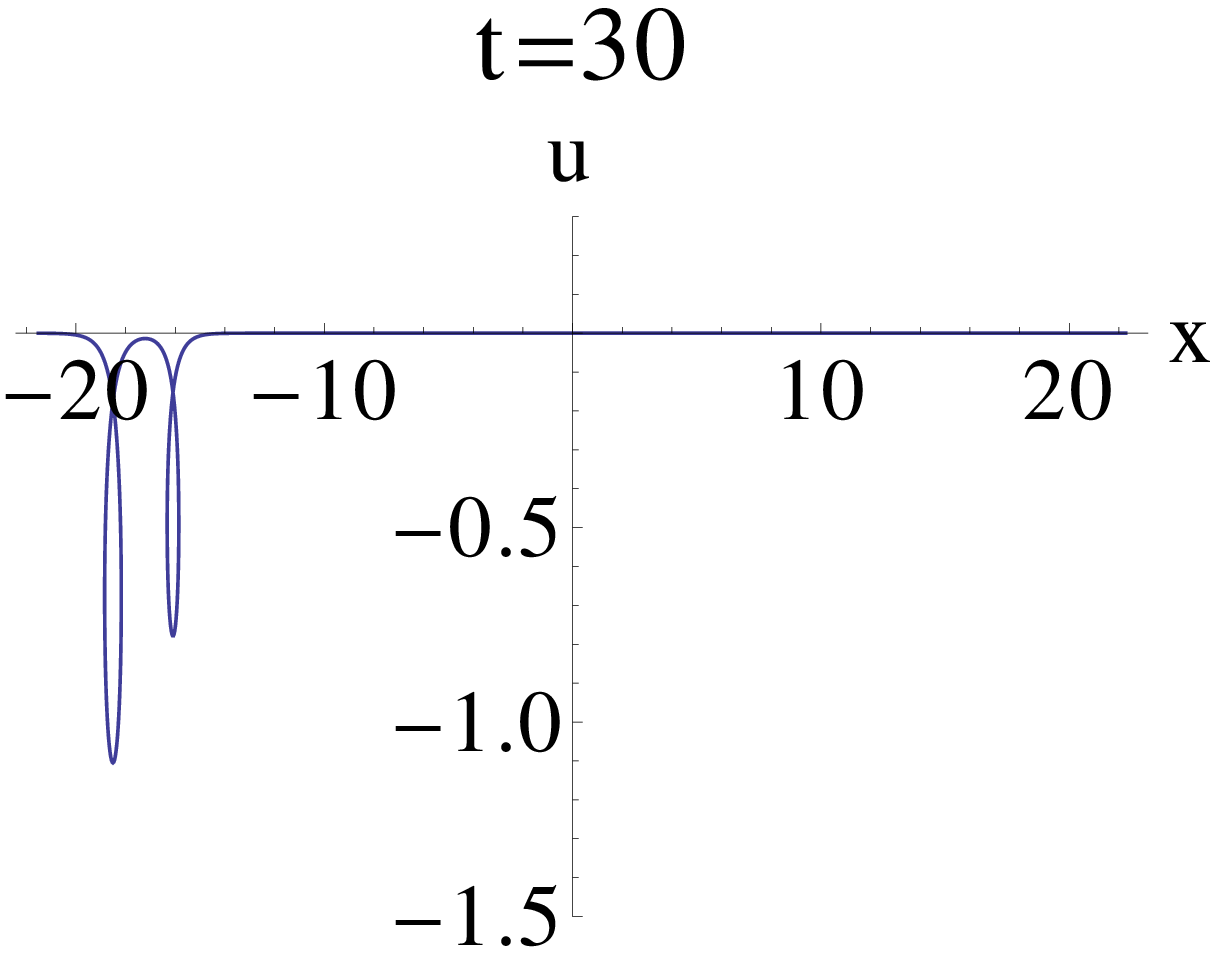}
\end{center}
\caption{2-loop soliton interaction. 
$\kappa=1$, $k_1=-\frac{7}{3}$, $k_2=\frac{8}{3}$.}
\end{figure}

\begin{figure}[!htbp]
\begin{center}
\includegraphics[width=6cm]{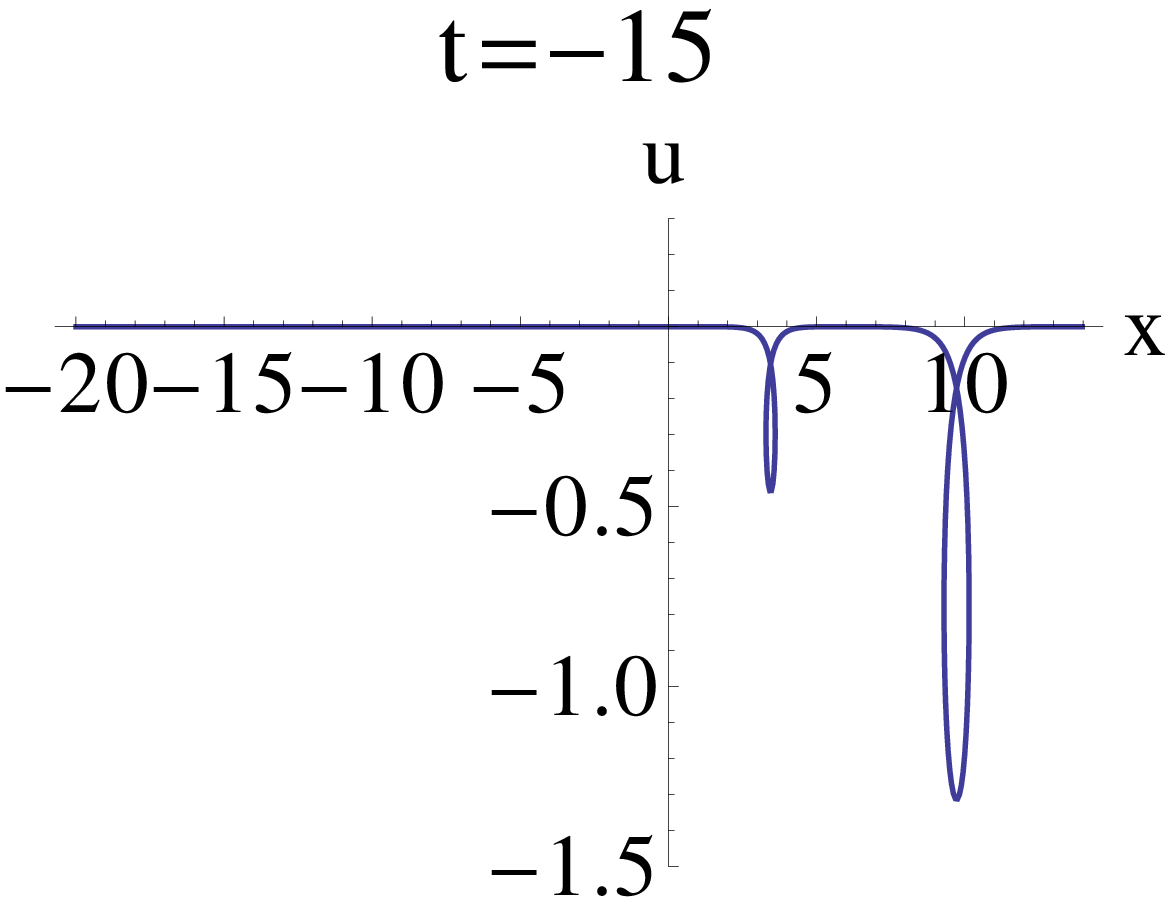}
\includegraphics[width=6cm]{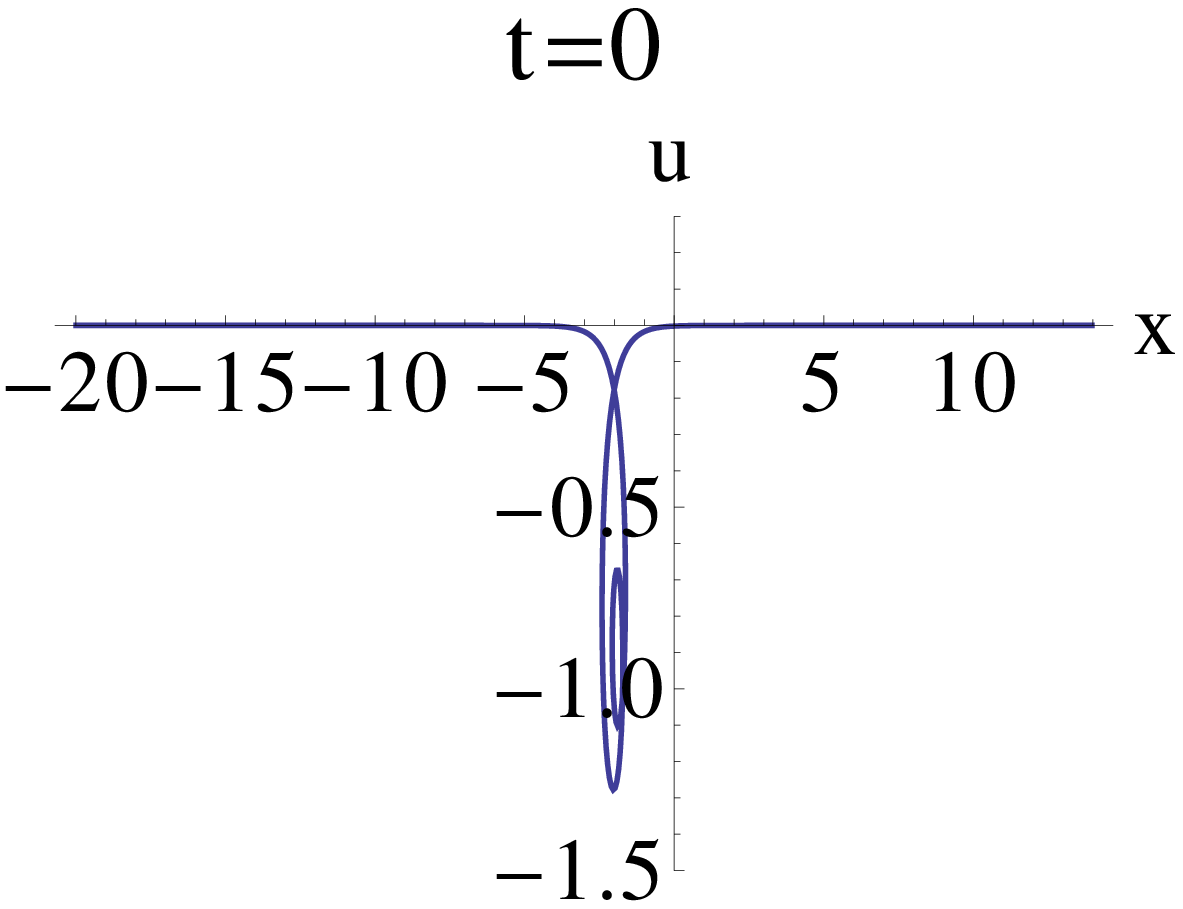}\\
\includegraphics[width=6cm]{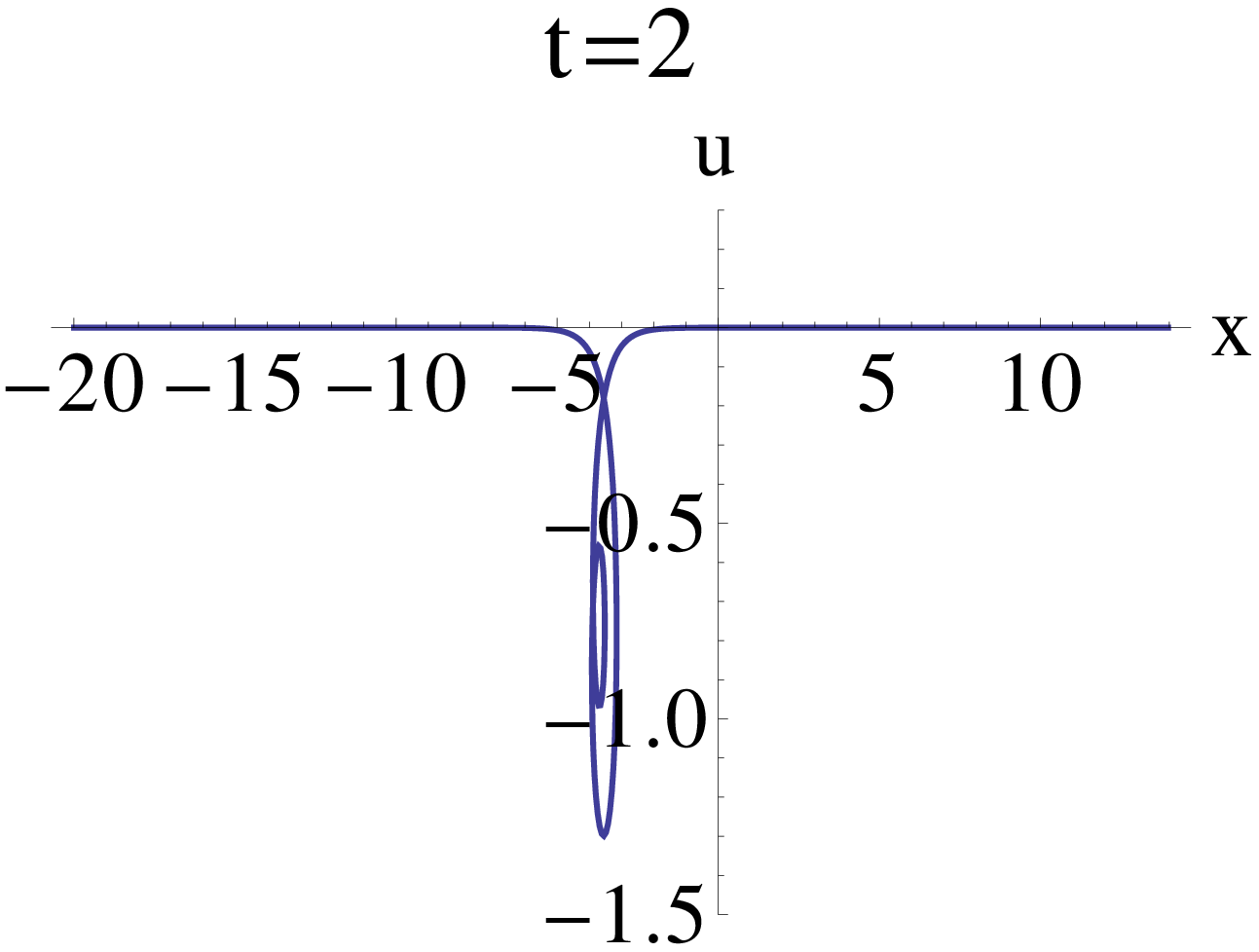}
\includegraphics[width=6cm]{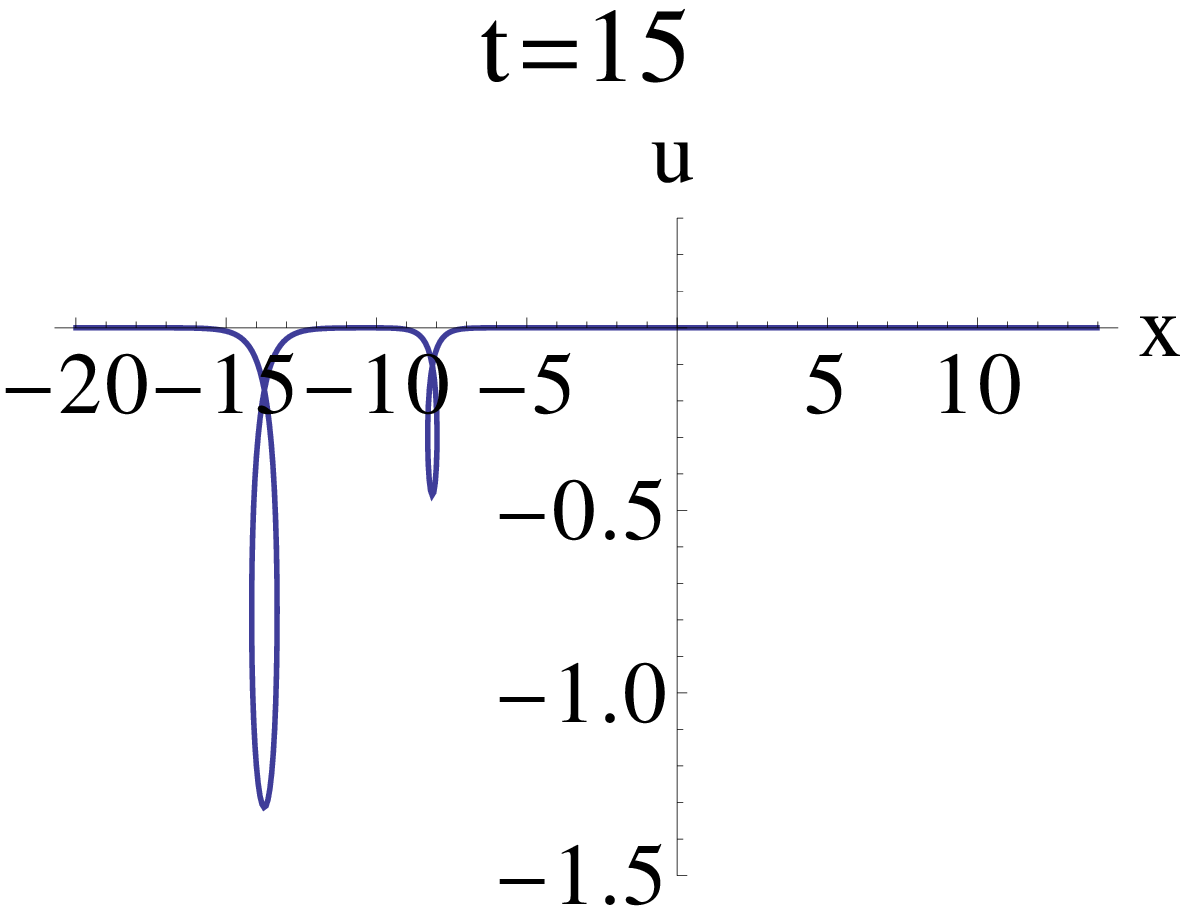}
\end{center}
\caption{2-loop soliton interaction. 
$\kappa=1$, $k_1=\frac{11}{5}$, $k_2=\frac{10}{3}$.}
\end{figure}

\begin{figure}[!htbp]
\begin{center}
\includegraphics[width=6cm]{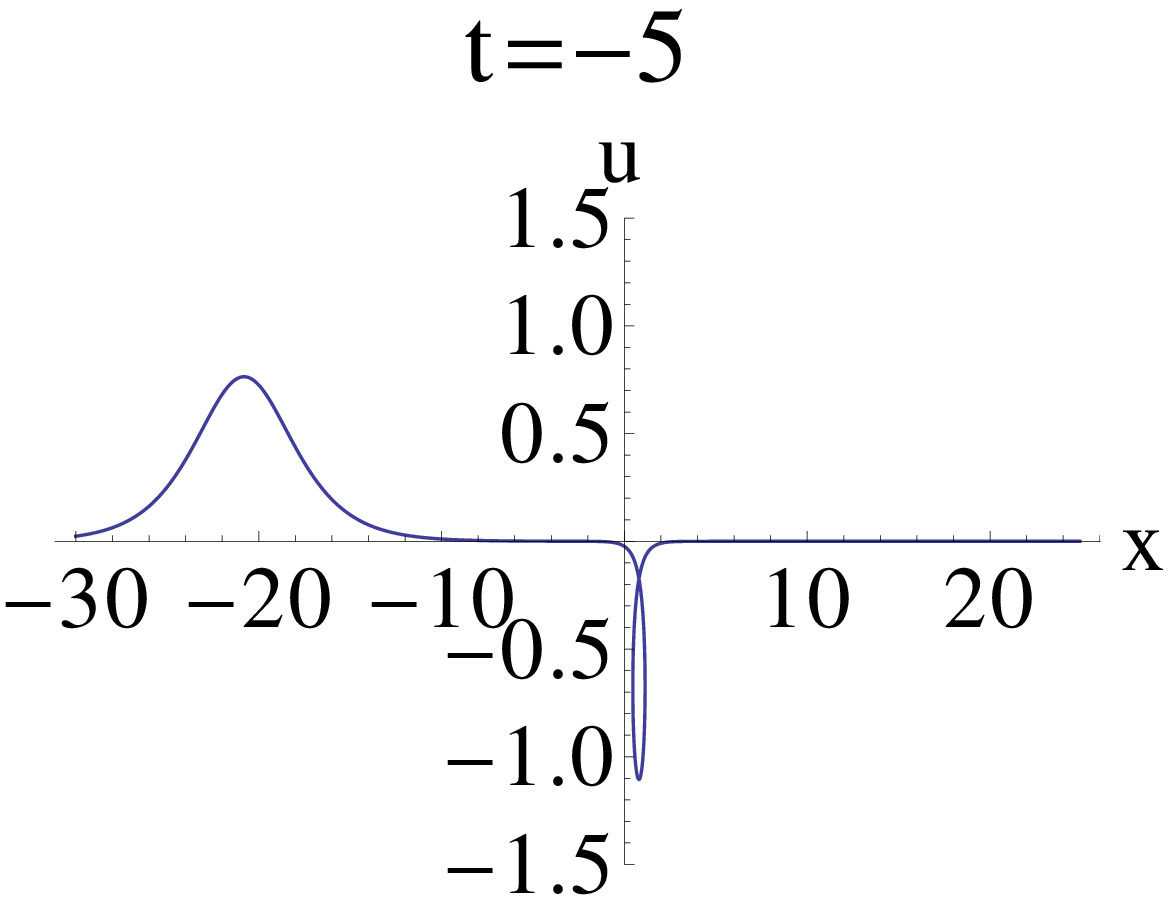}
\includegraphics[width=6cm]{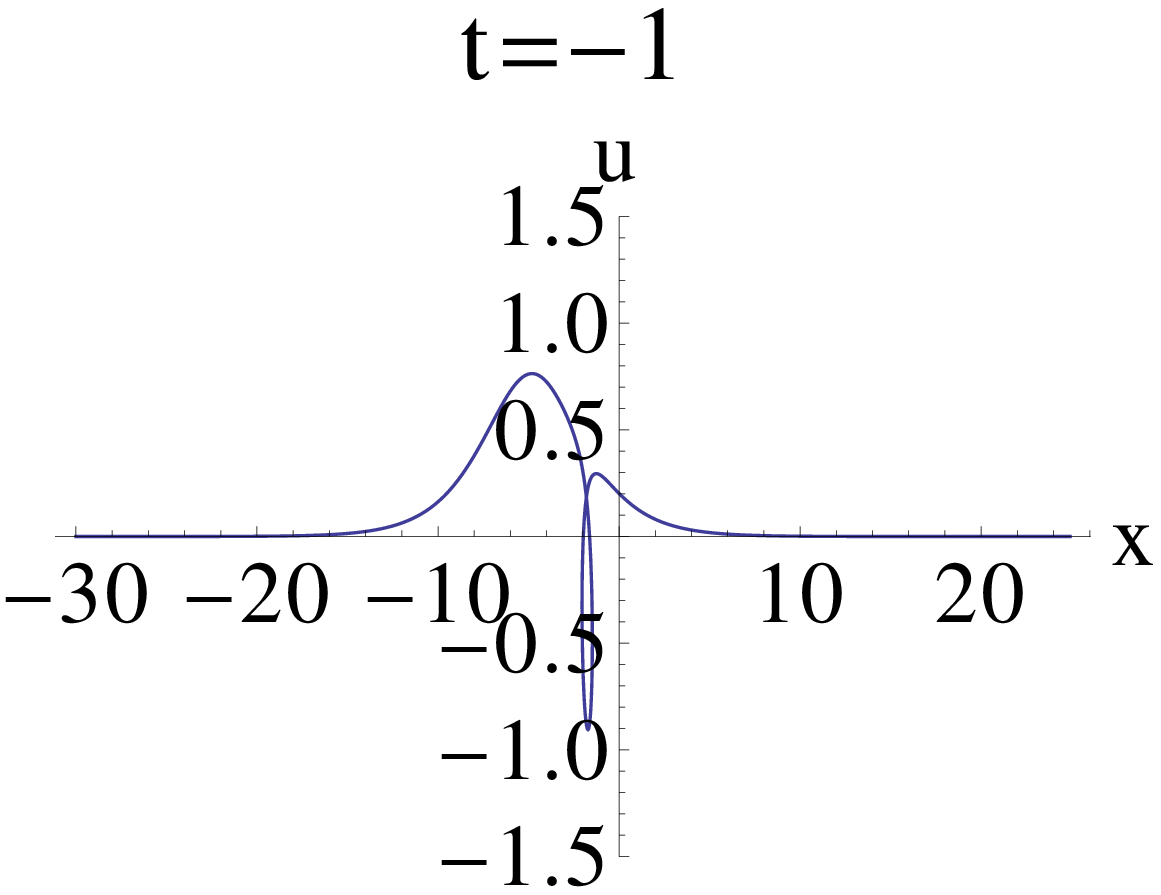}\\
\includegraphics[width=6cm]{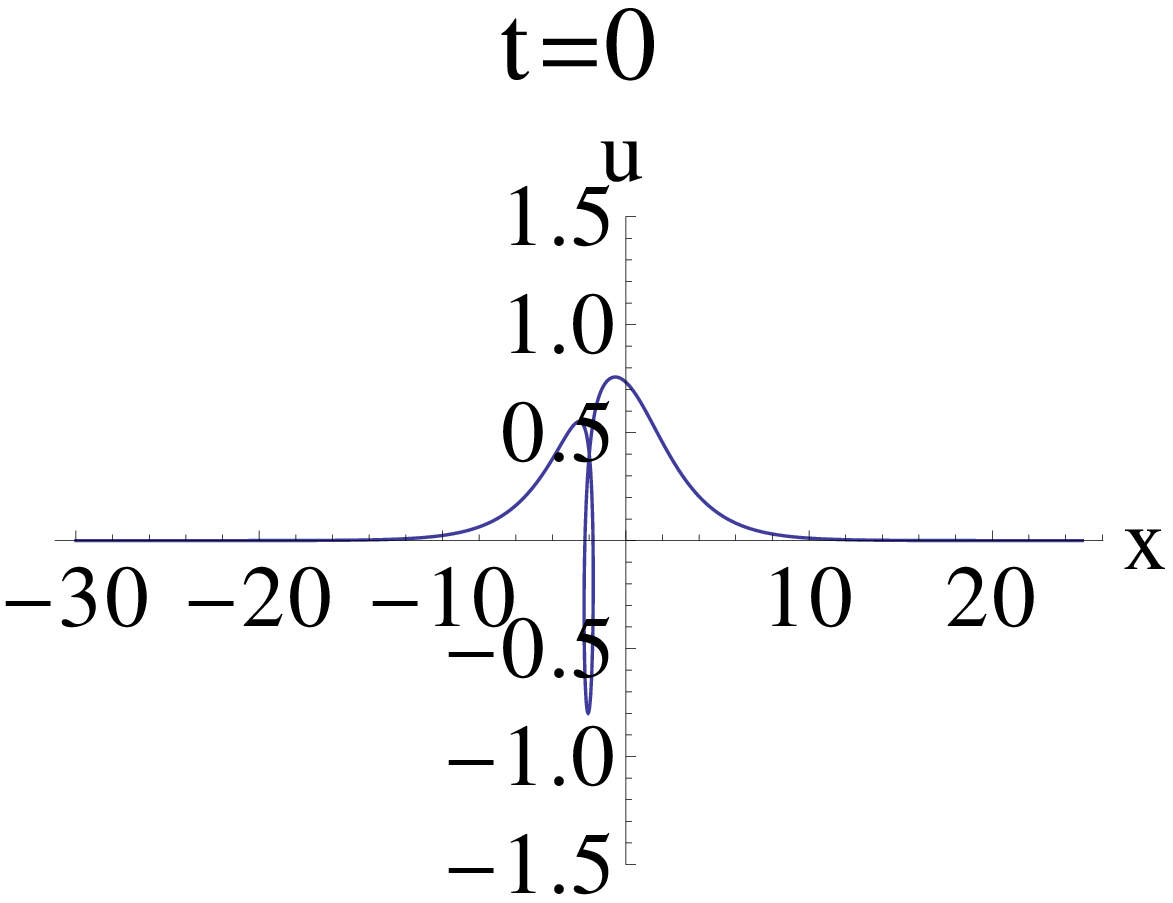}
\includegraphics[width=6cm]{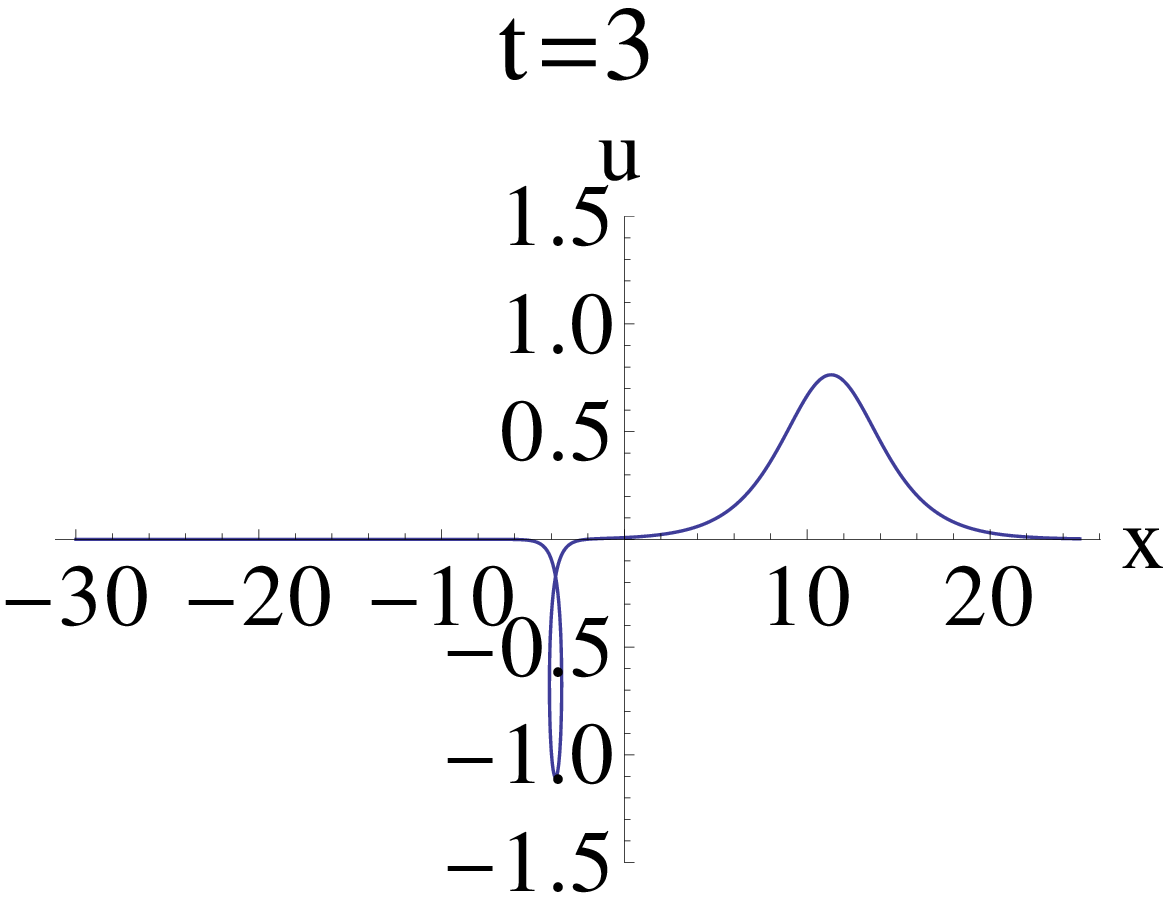}
\end{center}
\caption{Soliton and loop-soliton 
interaction. $\kappa=1$, $k_1=\frac{7}{3}$, $k_2=\frac{1}{2}$.}
\end{figure}

\section{Conclusions}

The DP equation is investigated from the point of view of 
determinant-pfaffian identities. 
We have established the reciprocal link between 
the DP equation and the pseudo 3-reduction of the  
$C_{\infty}$ two-dimensional Toda system and investigated 
the determinant-pfaffian identities 
(i.e., the relations of $\tau$-functions). 
We have shown that the $\tau$-functions of the DP equation satisfy 
the identities of determinants and pfaffians. 
The result in this article is consistent with the one obtained by
Matsuno~\cite{Matsuno-DP1,Matsuno-DP2}. 
Although we have obtained the same $\tau$-functions to Matsuno, 
we have proved several relations of the DP's $\tau$-functions 
from the point of view of determinant-pfaffian identities 
and established a formulation which 
can be applied to the problem of 
integrable discretization of the DP equation.

The result in this paper is useful for constructing 
an integrable discrete analogue of the DP equation. 
For the CH equation, we constructed an integrable discrete analogue 
of the CH equation by discretizing the determinant solutions 
and bilinear equations of the CH equation. 
For the DP equation, we can construct an integrable discrete analogue 
of the DP equation by discretizing pfaffian solutions 
and identities of pfaffians which have been given in this paper. 
We will report the detail in our forthcoming paper. 

\section*{Appendix A}
\renewcommand{\theequation}{A.\arabic{equation}}
\setcounter{equation}{0}
Let $A=(a_{i,j})_{1\leq i,j \leq 2N}$ be a $2N\times 2N$ skew-symmetric
matrix, i.e. $a_{i,j}=-a_{j,i}$. 
The pfaffian of $A$, ${\rm pf}(A)$, is defined as follows:
\begin{eqnarray}
{\rm pf}(A)&=&{\rm pf}(a_{i,j})_{1\leq i,j\leq 2N}
=\left.
\begin{array}{ccccc}
| & a_{1,2} &a_{1,3} & \cdots & a_{1,2N}\\
  &         &a_{2,3} & \cdots & a_{2,2N}\\
  &         &        & \ddots & \vdots\\
  &         &        &        & a_{2N-1,2N}\\
\end{array}
\right|\nonumber\\
&=&\sum_{\sigma}{\rm sgn}(\sigma)
\prod_{k=1}^Na_{i_{2k-1},i_{2k}}\,,
\end{eqnarray}
where the summation is taken over all permutations 
\[
\sigma=\left(
\begin{array}{cccc}
1 & 2 & \cdots & n\\
i_1 & i_2 & \cdots & i_n
\end{array}
\right)\,, 
\]
satisfying 
$i_1<i_2, i_3<i_4, \cdots , i_{2N-1}<i_{2N}$ 
and 
$i_1<i_3<\cdots <i_{2N-1}$,  
and ${\rm sgn}(\sigma)$ denotes the parity of the permutation $\sigma$. 

The pfaffian can be computed recursively by
\begin{equation}
{\rm pf}(A)=\sum_{i=2}^{2N}(-1)^ia_{1i}{\rm pf}(A_{\hat{1}\hat{i}})\,,
\end{equation}
where $A_{\hat{1}\hat{i}}$ denotes the matrix $A$ with both the first
and $i$-th rows and columns removed. 

The determinant of a skew-symmetric matrix $A$ is the square of the
pfaffian of $A$:
\begin{equation}
{\rm det}(A)=[{\rm pf}(A)]^2\,.
\end{equation}

\section*{Appendix B}
\renewcommand{\theequation}{B.\arabic{equation}}
\setcounter{equation}{0}
For a bordered determinant, we have the following
identity: 
\begin{eqnarray}
\fl &&\left|
\begin{array}{cccccc}
\alpha_{1,1} & \alpha_{1,2} & \cdots & \alpha_{1,2N-1} & \alpha_{1,2N} & a_1 \\
\alpha_{2,1} & \alpha_{2,2} & \cdots & \alpha_{2,2N-1} & \alpha_{2,2N} & a_2 \\
\vdots & \vdots  & \ddots &\vdots & \vdots &\vdots  \\
\alpha_{2N,1} & \alpha_{2N,2}  & \cdots &
 \alpha_{2N,2N-1} & \alpha_{2N,2N} & a_{2N} \\
b_1 & b_2 & \cdots & b_{2N-1} & b_{2N} & \delta 
\end{array}
\right|
= \delta^{1-2N}
{\rm det}\left|\delta \alpha_{i,j}-a_ib_j
\right|_{1\leq i,j\leq 2N}\,,\label{formula1}
\end{eqnarray}
where $\delta \neq 0$.
This  is obtained by adding the $(2N+1)$-th row multiplied by
$-a_i/\delta$ to the $i$th row. 

\section*{Appendix C}
\renewcommand{\theequation}{C.\arabic{equation}}
\setcounter{equation}{0}
Let $A=(a_{i,j})_{1\leq i,j \leq 2N+2}$ be a $(2N+2)\times (2N+2)$
skew-symmetric matrix. 
Assume $a_{2N+1,2N+2}\neq 0$.  
Adding the $(2N+2)$-th row multiplied by $a_{i,2N+1}/a_{2N+1,2N+2}$ and the $(2N+1)$-th row
multiplied by $-a_{i,2N+2}/a_{2N+1,2N+2}$ to the $i$th row, we obtain 
\begin{eqnarray*}
\fl  {\rm det}(A)&=&{\rm det}(a_{i,j})_{1\leq i,j\leq 2N+2}\\
\fl  &=& (a_{2N+1,2N+2})^{2-2N}{\rm
 det}(a_{2N+1,2N+2}\,a_{i,j}-
a_{2N+1,i}\,a_{2N+2,j}+a_{2N+2,i}\,a_{2N+1,j})\,,
\end{eqnarray*}
and 
\begin{eqnarray*}
\fl  &&{\rm pf}(A)={\rm pf}(a_{i,j})_{1\leq i,j\leq 2N+2}\\
\fl  &&\quad = (a_{2N+1,2N+2})^{1-N}{\rm
 pf}(a_{2N+1,2N+2}\,a_{i,j}-
a_{2N+1,i}\,a_{2N+2,j}+a_{2N+2,i}\,a_{2N+1,j})\,,
\end{eqnarray*}
(See e.g. \cite{Satake}.)
Thus we have the formulae
\begin{eqnarray*}
\fl &&\left|
\begin{array}{ccccccc}
0 & \alpha_{1,2} & \cdots & \alpha_{1,2N-1} & \alpha_{1,2N} & a_1 & b_1 \\
-\alpha_{1,2} & 0 & \cdots & \alpha_{2,2N-1} & \alpha_{2,2N} & a_2 & b_2\\
\vdots & \vdots  & \ddots &\vdots & \vdots &\vdots &\vdots \\
 -\alpha_{1,2N-1} & -\alpha_{2,2N-1}  & \cdots & 0 & \alpha_{2N-1,
  2N} & a_{2N-1} & b_{2N-1} \\
-\alpha_{1,2N} & -\alpha_{2,2N}  & \cdots &
 -\alpha_{2N-1,2N} & 0 & a_{2N} & b_{2N} \\
-a_1 & -a_2 & \cdots &-a_{2N-1} &-a_{2N} & 0 & \delta\\
-b_1 & -b_2 & \cdots & -b_{2N-1} & -b_{2N} & -\delta & 0
\end{array}
\right|\\
\fl &&\quad = \delta^{2-2N} 
\left|
\begin{array}{cccccc}
0 & \beta_{1,2} & \cdots & \beta_{1,2N-1} & \beta_{1,2N} \\
-\beta_{1,2} & 0 & \cdots & \beta_{2,2N-1} & \beta_{2,2N}\\
\vdots & \vdots  & \ddots &\vdots & \vdots \\
 -\beta_{1,2N-1} & -\beta_{2,2N-1}  & \cdots & 0 & \beta_{2N-1, 2N}\\
-\beta_{1,2N} & -\beta_{2,2N}  & \cdots &
 -\beta_{2N-1,2N} & 0
\end{array}
\right|\,,
\end{eqnarray*}
and 
\begin{eqnarray}
&& \left.
\begin{array}{ccccccc}
|  & \alpha_{1,2} & \alpha_{1,3} & \cdots & \alpha_{1,2N} & a_1 & b_1 \\
  &  & \alpha_{2,3} & \cdots & \alpha_{2,2N} & a_2 & b_2\\
  &  &  &\ddots & \vdots &\vdots &\vdots \\
   &  & &  & \alpha_{2N-1,
  2N} & a_{2N-1} & b_{2N-1} \\
  &   &  &
  &  & a_{2N} & b_{2N} \\
  &  &  &  &  &  & \delta\\
  &  &  &  & &  & 
\end{array}
\right|\nonumber\\
&& \qquad =
\delta^{1-N}
\left.
\begin{array}{cccccc}
| & \beta_{1,2} & \beta_{1,3} & \cdots & \beta_{1,2N} \\
 &  & \beta_{2,2} & \cdots & \beta_{2,2N}\\
&  &  &\ddots & \vdots \\
  &  & &  & \beta_{2N-1, 2N}\\
 &   &  &
 & 
\end{array}
\right|\,,\label{formula2}
\end{eqnarray}
where $\beta_{i,j}=\delta \alpha_{i,j}-a_ib_j+a_jb_i$ and $\delta\neq
0$.  

\section*{Appendix D}
\renewcommand{\theequation}{D.\arabic{equation}}
\setcounter{equation}{0}
For any determinant $A$ of order $n$, we have the Jacobi identity
\begin{equation}
A_{ij}A_{pq}-A_{iq}A_{pj}=AA_{ip,jq}\,,\label{Jacobi}
\end{equation}
where 
$A_{ij}$ is a first minor which can be obtained by deleting the $i$th row
and the $j$th column from $A$, and 
$A_{ip,jq}$ is a second minor which can be obtained by deleting the
$i$th and $p$th rows and the $j$th and $q$th columns.   
Setting $r=i=j$ and $s=p=q$, we have 
\begin{equation}
A_{rr}A_{ss}-A_{rs}A_{sr}=AA_{rs,rs}\,.
\end{equation}

Let $A$ be a skew-symmetric determinant. 
If $n$ is even, then $A_{rr}=A_{ss}=0$ and $A_{rs}=-A_{sr}$. 
Therefore we obtain
\begin{equation}
A_{rs}^2=AA_{rs,rs}\,,
\end{equation}
which leads to 
\begin{equation}
 A_{rs}={\rm pf}(A){\rm pf}(A_{rs,rs})\,. \label{det-paff-formula}
\end{equation}

If $n$ is odd, then $A=0$ and $A_{rs}=A_{sr}$. 
Therefore we obtain
\begin{equation}
A_{rs}^2=A_{rr}A_{ss}\,,
\end{equation}
which leads to 
\begin{equation}
A_{rs}={\rm pf}(A_{rr}){\rm pf}(A_{ss})\,.
\end{equation}
(See \cite{Muir,Vein-Dale}.)

Let
\begin{equation*}
\fl  A=\left|
\begin{array}{ccccccc}
0 & \alpha_{1,2}  &\cdots & \alpha_{1,2N-1} &
 \alpha_{1,2N} & a_1 & b_1 \\
-\alpha_{1,2} & 0  & \cdots & \alpha_{2,2N-1} &
 \alpha_{2,2N} & a_2 & b_2\\
\vdots & \vdots  & \ddots &\vdots & \vdots &\vdots \\
 -\alpha_{1,2N-1} & -\alpha_{2,2N-1} & \cdots & 0 &
  \alpha_{2N-1, 2N} & a_{2N-1} & b_{2N-1} \\
-\alpha_{1,2N} & -\alpha_{2,2N}  & \cdots &
 -\alpha_{2N-1,2N} & 0 & a_{2N} & b_{2N} \\
-a_1 & -a_2  & \cdots & -a_{2N-1} & -a_{2N} & 0 &\delta \\
-b_1 & -b_2  & \cdots & -b_{2N-1} & -b_{2N} & -\delta & 0 
\end{array}
\right|\,. 
\end{equation*}
Using (\ref{det-paff-formula}) and (\ref{formula2}), we obtain the formula
\begin{eqnarray}
\fl &&\left|
\begin{array}{ccccccc}
0 & \alpha_{1,2} & \alpha_{1,3} &\cdots & \alpha_{1,2N-1} & \alpha_{1,2N} & b_1 \\
-\alpha_{1,2} & 0 & \alpha_{2,3} & \cdots & \alpha_{2,2N-1} & \alpha_{2,2N} & b_2\\
-\alpha_{1,3} & -\alpha_{2,3} & 0 & \cdots & \alpha_{3,2N-1} & \alpha _{3,2N} & b_3\\
\vdots & \vdots & \ddots & \ddots &\vdots & \vdots &\vdots \\
 -\alpha_{1,2N-1} & -\alpha_{2,2N-1} & -\alpha_{3,2N-1} & \cdots & 0 & \alpha_{2N-1, 2N} & b_{2N-1} \\
-\alpha_{1,2N} & -\alpha_{2,2N} & -\alpha_{3,2N} & \cdots &
 -\alpha_{2N-1,2N} & 0 & b_{2N} \\
-a_1 & -a_2 & -a_3 & \cdots &-a_{2N-1} &-a_{2N} & \delta
\end{array}
\right|\nonumber\\
\fl &&\qquad =  {\rm pf}(\alpha_{i,j})_{1\leq i,j\leq 2N}\,\left.
\begin{array}{cccccccc}
| & \alpha_{1,2} & \alpha_{1,3} &\cdots & 
 \alpha_{1,2N} & a_1 & b_1 \\
&  & \alpha_{2,3} & \cdots  & \alpha_{2,2N} & a_2 & b_2\\
&  &  & \ddots  &\vdots & \vdots &\vdots \\
 &  & &  &   \alpha_{2N-1, 2N} & a_{2N-1} & b_{2N-1} \\
 &  & &  &        & a_{2N} & b_{2N} \\
&  &  & & &    & \delta
\end{array}
\right|\,.\label{formula3}
\end{eqnarray}

\section*{Appendix E}
\renewcommand{\theequation}{E.\arabic{equation}}
\setcounter{equation}{0}
Let
\begin{equation*}
\fl  A=
\left|
\begin{array}{ccccccc}
0 & \alpha_{1,2}  &\cdots & \alpha_{1,2N-1} &
 \alpha_{1,2N} & a_1 & b_1 \\
-\alpha_{1,2} & 0  & \cdots & \alpha_{2,2N-1} &
 \alpha_{2,2N} & a_2 & b_2\\
\vdots & \vdots  & \ddots &\vdots & \vdots &\vdots \\
 -\alpha_{1,2N-1} & -\alpha_{2,2N-1} & \cdots & 0 &
  \alpha_{2N-1, 2N} & a_{2N-1} & b_{2N-1} \\
-\alpha_{1,2N} & -\alpha_{2,2N}  & \cdots &
 -\alpha_{2N-1,2N} & 0 & a_{2N} & b_{2N} \\
c_1 & c_2  & \cdots &c_{2N-1} &c_{2N} & \alpha &\beta \\
d_1 & d_2  & \cdots &d_{2N-1} &d_{2N} & \gamma &\delta 
\end{array}
\right|\,.
\end{equation*}
Using the Jacobi identity (\ref{Jacobi}), we obtain 
\begin{eqnarray*}
&&\left|
\begin{array}{ccccccc}
0 & \alpha_{1,2}  &\cdots & \alpha_{1,2N-1} &
 \alpha_{1,2N} & a_1 & b_1 \\
-\alpha_{1,2} & 0  & \cdots & \alpha_{2,2N-1} &
 \alpha_{2,2N} & a_2 & b_2\\
\vdots & \vdots  & \ddots &\vdots & \vdots &\vdots \\
 -\alpha_{1,2N-1} & -\alpha_{2,2N-1} & \cdots & 0 &
  \alpha_{2N-1, 2N} & a_{2N-1} & b_{2N-1} \\
-\alpha_{1,2N} & -\alpha_{2,2N}  & \cdots &
 -\alpha_{2N-1,2N} & 0 & a_{2N} & b_{2N} \\
c_1 & c_2  & \cdots &c_{2N-1} &c_{2N} & \alpha &\beta \\
d_1 & d_2  & \cdots &d_{2N-1} &d_{2N} & \gamma &\delta 
\end{array}
\right|\\
&&\times 
\left|
\begin{array}{ccccc}
0 & \alpha_{1,2}  &\cdots & \alpha_{1,2N-1} &
 \alpha_{1,2N} \\
-\alpha_{1,2} & 0  & \cdots & \alpha_{2,2N-1} &
 \alpha_{2,2N} \\
\vdots & \vdots  & \ddots &\vdots  \\
 -\alpha_{1,2N-1} & -\alpha_{2,2N-1} & \cdots & 0 &
  \alpha_{2N-1, 2N}  \\
-\alpha_{1,2N} & -\alpha_{2,2N}  & \cdots &
 -\alpha_{2N-1,2N} & 0  
\end{array}
\right|\\
&&=\left|
\begin{array}{cccccc}
0 & \alpha_{1,2}  &\cdots & \alpha_{1,2N-1} &
 \alpha_{1,2N} & a_1 \\
-\alpha_{1,2} & 0  & \cdots & \alpha_{2,2N-1} &
 \alpha_{2,2N} & a_2 \\
\vdots & \vdots  & \ddots &\vdots & \vdots  \\
 -\alpha_{1,2N-1} & -\alpha_{2,2N-1} & \cdots & 0 &
  \alpha_{2N-1, 2N} & a_{2N-1}  \\
-\alpha_{1,2N} & -\alpha_{2,2N}  & \cdots &
 -\alpha_{2N-1,2N} & 0 & a_{2N}  \\
c_1 & c_2  & \cdots &c_{2N-1} &c_{2N} & \alpha 
\end{array}
\right|\\
&&\times \left|
\begin{array}{cccccc}
0 & \alpha_{1,2}  &\cdots & \alpha_{1,2N-1} &
 \alpha_{1,2N} & b_1 \\
-\alpha_{1,2} & 0  & \cdots & \alpha_{2,2N-1} &
 \alpha_{2,2N} & b_2 \\
\vdots & \vdots  & \ddots &\vdots & \vdots  \\
 -\alpha_{1,2N-1} & -\alpha_{2,2N-1} & \cdots & 0 &
  \alpha_{2N-1, 2N} & b_{2N-1}  \\
-\alpha_{1,2N} & -\alpha_{2,2N}  & \cdots &
 -\alpha_{2N-1,2N} & 0 & b_{2N}  \\
d_1 & d_2  & \cdots &d_{2N-1} &d_{2N} & \delta 
\end{array}
\right|\\
&&- \left|
\begin{array}{cccccc}
0 & \alpha_{1,2}  &\cdots & \alpha_{1,2N-1} &
 \alpha_{1,2N} & b_1 \\
-\alpha_{1,2} & 0  & \cdots & \alpha_{2,2N-1} &
 \alpha_{2,2N} & b_2 \\
\vdots & \vdots  & \ddots &\vdots & \vdots  \\
 -\alpha_{1,2N-1} & -\alpha_{2,2N-1} & \cdots & 0 &
  \alpha_{2N-1, 2N} & b_{2N-1}  \\
-\alpha_{1,2N} & -\alpha_{2,2N}  & \cdots &
 -\alpha_{2N-1,2N} & 0 & b_{2N}  \\
c_1 & c_2  & \cdots &c_{2N-1} &c_{2N} & \beta 
\end{array}
\right|\\
&&\times \left|
\begin{array}{cccccc}
0 & \alpha_{1,2}  &\cdots & \alpha_{1,2N-1} &
 \alpha_{1,2N} & a_1 \\
-\alpha_{1,2} & 0  & \cdots & \alpha_{2,2N-1} &
 \alpha_{2,2N} & a_2 \\
\vdots & \vdots  & \ddots &\vdots & \vdots  \\
 -\alpha_{1,2N-1} & -\alpha_{2,2N-1} & \cdots & 0 &
  \alpha_{2N-1, 2N} & a_{2N-1}  \\
-\alpha_{1,2N} & -\alpha_{2,2N}  & \cdots &
 -\alpha_{2N-1,2N} & 0 & a_{2N}  \\
d_1 & d_2  & \cdots &d_{2N-1} &d_{2N} & \gamma 
\end{array}
\right|\,.
\end{eqnarray*}
From (\ref{formula3}), we obtain
\begin{eqnarray}
 &&\left|
\begin{array}{ccccccc}
0 & \alpha_{1,2}  &\cdots & \alpha_{1,2N-1} &
 \alpha_{1,2N} & a_1 & b_1 \\
-\alpha_{1,2} & 0  & \cdots & \alpha_{2,2N-1} &
 \alpha_{2,2N} & a_2 & b_2\\
\vdots & \vdots  & \ddots &\vdots & \vdots &\vdots \\
 -\alpha_{1,2N-1} & -\alpha_{2,2N-1} & \cdots & 0 &
  \alpha_{2N-1, 2N} & a_{2N-1} & b_{2N-1} \\
-\alpha_{1,2N} & -\alpha_{2,2N}  & \cdots &
 -\alpha_{2N-1,2N} & 0 & a_{2N} & b_{2N} \\
c_1 & c_2  & \cdots &c_{2N-1} &c_{2N} & \alpha &\beta \\
d_1 & d_2  & \cdots &d_{2N-1} &d_{2N} & \gamma &\delta 
\end{array}
\right|\nonumber\\
 &&\quad =\left.
\begin{array}{cccccccc}
| & \alpha_{1,2} & \alpha_{1,3} &\cdots & 
 \alpha_{1,2N} & a_1 & c_1 \\
&  & \alpha_{2,3} & \cdots  & \alpha_{2,2N} & a_2 & c_2\\
&  &  & \ddots  &\vdots & \vdots &\vdots \\
 &  &  & &   \alpha_{2N-1, 2N} & a_{2N-1} & c_{2N-1} \\
 &  & &  &        & a_{2N} & c_{2N} \\
&  &  & & &    & \alpha
\end{array}
\right|\nonumber
\\
&&\quad \times 
\left.
\begin{array}{cccccccc}
| & \alpha_{1,2} & \alpha_{1,3} &\cdots & 
 \alpha_{1,2N} & b_1 & d_1 \\
&  & \alpha_{2,3} & \cdots  & \alpha_{2,2N} & b_2 & d_2\\
&  &  & \ddots  &\vdots & \vdots &\vdots \\
 &  &  & &   \alpha_{2N-1, 2N} & b_{2N-1} & d_{2N-1} \\
 &  & &  &        & b_{2N} & d_{2N} \\
&  &  & & &    & \delta
\end{array}
\right|\nonumber\\
 &&\quad -
\left.
\begin{array}{cccccccc}
| & \alpha_{1,2} & \alpha_{1,3} &\cdots & 
 \alpha_{1,2N} & b_1 & c_1 \\
&  & \alpha_{2,3} & \cdots  & \alpha_{2,2N} & b_2 & c_2\\
&  &  & \ddots  &\vdots & \vdots &\vdots \\
 &  &  & &   \alpha_{2N-1, 2N} & b_{2N-1} & c_{2N-1} \\
 &  & &  &        & b_{2N} & c_{2N} \\
&  &  & & &    & \beta
\end{array}
\right|\nonumber\\
&&\quad \times 
\left.
\begin{array}{cccccccc}
| & \alpha_{1,2} & \alpha_{1,3} &\cdots & 
 \alpha_{1,2N} & a_1 & d_1 \\
&  & \alpha_{2,3} & \cdots  & \alpha_{2,2N} & a_2 & d_2\\
&  &  & \ddots  &\vdots & \vdots &\vdots \\
 &  &  & &   \alpha_{2N-1, 2N} & a_{2N-1} & d_{2N-1} \\
 &  & &  &        & a_{2N} & d_{2N} \\
&  &  & & &    & \gamma
\end{array}
\right|\,.\label{formula4}
%%% &&=
%%%\alpha^{1-N}\delta^{1-N}{\rm pf}
%%%(\alpha \alpha_{i,j}-a_ic_j+a_jc_i)_{1\leq i,j\leq 2N}\,
%%%{\rm pf}(\delta\alpha_{i,j}-b_id_j+b_jd_i)_{1\leq i,j\leq 2N}\nonumber\\
%%%&&\quad -\beta^{1-N}\gamma^{1-N}{\rm pf}(\beta\alpha_{i,j}
%%%-b_ic_j+b_jc_i)_{1\leq i,j\leq 2N}\,
%%%{\rm pf}(\gamma \alpha_{i,j}-a_id_j+a_jd_i)_{1\leq i,j\leq
%%%2N}\,.\nonumber\\
\end{eqnarray}

%%\newpage
%%%%%%%%%%%%%%%%%%%%%%%%%%%%%%%%%%%%%%%%%%%%%%%%%%%%%%%%%%%%%%%%%%%%%%%%%

\end{document}